\documentclass{fundam}

\publyear{25}
\papernumber{8}
\volume{195}
\issue{1-4}
\theDOI{10.46298/fi.12479}

\versionForARXIV

\usepackage[utf8]{inputenc}
\usepackage{latexsym,amssymb}
\usepackage{normalforms.macros}

\usepackage{tikz-cd}
\usetikzlibrary{arrows.meta}
\tikzcdset{arrow style=tikz, diagrams={>=Stealth}}

\title{
	Normal Forms for Elements of ${}^*$-Continuous Kleene Algebras
        Representing the Context-Free Languages
}

\author{%
        Mark Hopkins
\\	The Federation Archive
\\	https://github.com/FederationArchive
\\	federation2005@netzero.net
\\	\and
	Hans Leiß
\\      Centrum f\"ur Informations- und Sprachverarbeitung,
\\      Ludwig-Maxi\-mi\-li\-ans-Universit\"at M\"unchen (retired)
\\	h.leiss@gmx.de
}
\date{24 April 2025} 

\begin{document}
\maketitle
\runninghead{M.~Hopkins, H.~Lei\ss}{Normal forms for elements of $K\xR C_2'$}

\begin{abstract}
Within the tensor product $K \mathop{\otimes_{\cal R}} C_2'$ of any \Star-continuous
Kleene algebra $K$ with the polycyclic \Star-continuous Kleene algebra $C_2'$ over
two bracket pairs there is a copy of the fixed-point closure of $K$: the centralizer
of $C_2'$ in $K \mathop{\otimes_{\cal R}} C_2'$.
Using an automata-theoretic representation of elements of $K\mathop{\otimes_{\cal R}}
C_2'$ \`a la Kleene and with the aid of normal form theorems that restrict the
occurrences of brackets on paths through the automata, we develop a foundation for a
calculus of context-free expressions without variable binders.
We also give some results on the bra-ket *-continuous Kleene algebra $C_2$, motivate
the ``completeness equation'' that distinguishes $C_2$ from $C_2'$, and show that
$C_2'$ validates a relativized form of this equation.
\end{abstract}

\section{Introduction}

A Kleene algebra $K = (K,+,\cdot,{}^*,0,1)$ is \Star-continuous if
\[
	a \cdot c^*\cdot b = \Sum{a\cdot c^n\cdot b}{n\in\N}
\]
for all $a, b, c \in K$, where $\sum$ is the least upper bound with respect to the natural partial order $\leq$ on $K$ given by $a \leq b$ iff $a + b = b$.
Well-known examples of \Star-continuous Kleene algebras are the algebras $\CR M = (\CR M, +,\cdot,{}^*,0,1)$
of regular or ``rational'' subsets of a monoid $M = (M,\cdot^M,1^M)$, where
$0 := \emptyset$, $1 := \{1^M\}$ and $+$ is union, $\cdot$ is elementwise product,
and \Star\ is iteration or ``monoid closure'', i.e.~for $A \in \CR M$, $A^*$ is the least $B \supseteq A$ that contains $1^M$ and is closed under $\cdot^M$.

We will make use of two other kinds of \Star-continuous Kleene algebras:
quotients $K/\rho$ of \Star-continuous Kleene algebras $K$ under $\CR$-congruences $\rho$ on $K$,
i.e.~semiring congruences which make suprema of regular subsets
congruent if their elements are congruent in a suitable sense, and tensor products $K\xR K'$ of \Star-continuous Kleene algebras $K$ and $K'$.

Let $\Delta_m$ be a set of $m$ pairs of ``brackets'', $p_i, q_i$, $i < m$, and $\CR\Delta_m^*$ the \Star-continuous Kleene algebra of regular subsets of $\Delta_m^*$.
Hopkins \cite{Hopkins-I-2008} considers the $\CR$-congruence $\rho_m$ on $\CR \Delta_m^*$ generated by the equation set
\begin{equation}\label{bra-ket-Cn}
	\setof{p_iq_j=\delta_{i,j}}{i,j<m} \cup \{q_0p_0+ \ldots + q_{m-1}p_{m-1}= 1\}
\end{equation}
and the finer $\CR$-congruence $\rho_m'$ generated by the equations
\begin{equation}\label{poly-Cn}
	\setof{p_iq_j = \delta_{i,j}}{i,j<m},
\end{equation}
where $\delta_{i,j}$ is the Kronecker $\delta$.
The latter equations allow us to algebraically distinguish matching brackets, where $p_iq_j=1$, from non-matching ones, where $p_iq_j=0$.\footnote{
	~In $\CR\Delta^*_m$, elements of $\Delta_m^*$ are interpreted by their singleton sets, 0 by the empty set.
}
\,These $\CR$-congruences give rise to the \blue{bra-ket}
and the \blue{polycyclic} \Star-con\-tinuous Kleene algebra $\blue{C_m} = \CR\Delta_m^*/\rho_m$ and $\blue{C_m'} = \CR\Delta_m^*/\rho_m'$, respectively.
For $m > 2$, $C_m$ can be coded in $C_2$ and $C_m'$ in $C_2'$, so we focus on the case $m=2$.

Two \Star-continuous Kleene algebras $K$ and $C$ can be combined to a \blue{tensor product $K\xR C$} which, intuitively,
is the smallest common \Star-continuous Kleene algebra extension of $K$ and $C$ in which elements of $K$ commute with those of $C$.

In unpublished work, the first author showed that for any \Star-continuous Kleene algebra $K$, the tensor product $K\xR C_2$ contains an isomorphic copy of the fixed-point closure of $K$.
In particular, for finite alphabets $X$, each context-free set $L \subseteq X^*$ is represented in $\CR X^*\xR C_2$ as the value of a regular expression over the disjoint union $X \dotcup \Delta_2$ of $X$ and $\Delta_2$.
In fact, the \blue{centralizer of $C_2$} in $K\xR C_2$, i.e.~the set of those elements of $K\xR C_2$ that commute with every element of $C_2$,
consists of exactly the representations of context-free subsets of the multiplicative monoid of $K$.
These results constitute a generalization of the Chomsky and Schützenberger
representation theorem\,(\cite{Chomsky63}, Proposition 2) in formal language theory,
which says that any context-free set $L\subseteq X^*$ is the image $h(R\cap D)$ of a
regular set $R\subseteq (X\cup\Delta)^*$ under a homomorphism $h:(X\cup\Delta)^*\to
X^*$ that keeps elements of $X$ fixed and ``erases'' symbols of $\Delta$ to 1.
The generalization is shown  in \cite{Leiss22} with the simpler algebra $K\xR C_2'$ instead of $K\xR C_2$.

It is therefore of some interest to understand the structure of $K\xR C_2$ and $K\xR C_2'$.
In this article, an extension of \cite{HopkinsLeiss23}, we focus on $K\xR C_2'$,
using ideas from and improvements of unpublished results on $K\xR C_2$ by the first
author. Our main results are normal forms for elements of $K\xR C_2'$ that relate
arbitrary elements to those of the centralizer of $C_2'$.
We also present some results specific to $C_2$ and its matrix algebra.
The rest of this article is structured as follows.

Section \ref{sec-KA} recalls the definitions of \Star-continuous Kleene algebras (aka
$\CR$-dioids), bra-ket and polycyclic \Star-continuous Kleene algebras, and quotients
and tensor products of \Star-continuous Kleene algebras.  We then show a Kleene representation theorem,
i.e.~that each element $\phi$ of $K\xR C_2'$ is the value $L(\CA)=SA^*F$ of an
automaton $\CA=\Z{S, A, F}$, where $S \in \{0, 1\}^{1\times n}$ resp.~$F \in \{0,
1\}^{n\times 1}$ code the set of initial resp.~accepting states of the $n$ states of
$\CA$ and $A \in \Mat nn{K\xR C_2'}$ is a transition matrix.

Section \ref{sec-Normalforms} refines the representation $\phi=L(\CA)$ to a \emph{normal form} where brackets on paths through the automaton $\CA$ occur mostly in a balanced way.
Section 3.1 identifies, in any Kleene algebra with elements $u, x, v$, the value $(u+x+v)^*$ with the value $(Nv)^* N (uN)^*$,
provided the algebra has a least solution $N$ of the inequation $y \geq (x+uyv)^*$ defining Dyck's language $D(x) \subseteq \{u, x, v\}^*$ with ``bracket'' pair $u, v$.
We then show that for any \Star-continuous Kleene algebra $K$ and $n \geq 1$, $\Mat
nn{K\xR C_2'}$ has such a solution $N$ of $y\geq (UyV+X)^*$ for matrices $U$ of 0's
and opening brackets from $C_2'$, $X$ of elements of $K$, and $V$ of 0's and closing
brackets from $C_2'$, and that entries of $N$ belong to the centralizer of $C_2'$ in
$K\xR C_2'$.

Section \ref{sec-normal-forms} refines the representation $\phi=L(\CA)$ to the
sketched normal form: the transition matrix $A$ can be split as $A=U+X+V$ into a
matrix $X \in K^{n\times n}$ of transitions by elements of $K$, a matrix $U \in \{0,
\p0, \p1\}^{n\times n}$ of transitions by $0$ or opening brackets of $C_2'$, and a
matrix $V \in \{0, \q0, \q1\}^{n\times n}$ of transitions by 0 or closing brackets of
$C_2'$.  Then $A^*$ can be normalized to $(NV)^*N (UN)^*$, where $N$ is balanced in
$U$ and $V$ and all other occurrences of closing brackets $V$ are in front of all
other occurrences of opening brackets $U$.  We call $SA^*F=S(NV)^*N(UN)^*F$ the first
normal form of $\phi$.  This result is a generalization of a normal form for elements
of the polycyclic monoid $P_2'[X]$, the quotient of $(\Delta_2\cup X\cup\{0\})^*$ by the monoid congruence generated by the bracket match- and mismatch
equations, the equations for commuting brackets of $\Delta_2$ with symbols of $X$, and
the annihilator equations for 0. Namely, if $\Delta_2= U \cup V$ is split into
opening brackets $U$ and closing brackets $V$, any $w\in (\Delta_2\cup X\cup\{0\})^*$
is congruent to a normal form $\nf{w}\in V^*X^*U^*\cup\{0\}$. (The
centralizer of $\Delta_2$ in $P_2'[X]$ is $X^*\cup\{0\}$, so the analogues of $N$ are
contracted in the factor $X^*$.)

Section \ref{sec-reduced-nf} proves a conjecture of \cite{HopkinsLeiss23}: if
$\phi=L(\CA)$ belongs to the centralizer of $C_2'$ in $K\xR C_2'$, then the normal
form $SA^*F=S(NV)^*N(UN)^*F$ can be simplified to $SA^*F = SNF$.  We call this the
reduced normal form.  For this, we have to assume that $K$ is non-trivial and has no
zero divisors, which is satisfied e.g.~when $K=\CR M$ for a monoid $M$.  A second
normal form is given for a slightly more general transition matrix $A$ than $U+X+V$,
which is useful for the representation of products of context-free subsets.
For the elements of the centralizer of $C_2'$ in $K\xR C_2'$ only, a different
characterization had been given in \cite{Leiss22}. The normal form theorems presented
here improve on this by showing how the elements of the centralizer of $C_2'$, i.e.~the
representations of context-free subsets of $K$ in $K\xR C_2'$, relate to the
remaining elements of $K\xR C_2'$.

For a finite set $X$, the elements of $\CR X^*\xR C_2'$ are named by regular
expressions over $\Delta_2\dotcup X$, as mentioned above. A subset of those, called
the context-free expressions over $X$, name the elements of the centralizer of $C_2'$
in $\CR X^*\xR C_2'$, i.e.~the representations of the context-free languages
$L\subseteq X^*$.
Section \ref{sec-combining-nfs} provides a foundation of a calculus of context-free
expressions by showing how to combine normal forms for elements of any $K\xR C_2'$ by
regular operations.

Section \ref{sec-bra-ket} deals with the bra-ket \Star-continuous Kleene algebras
$C_m$. Section \ref{sec-mat-Cm} gives an interpretation of $C_m$ in the algebra of
binary relations on a countably infinite set, $\Mat\omega\omega\B$. We also show that
$C_m$ is isomorphic to $\Mat mm{C_m}$ and $C_m\xR \Mat mm\B$, thereby excluding an
interpretation by finite-dimensional matrices.  Section \ref{sec-completeness}
considers a natural interpretation of brackets as stack operations, where $p_i$
pushes symbol $i\in\{1,\ldots,m-1\}$ to the stack and $q_i$ pops $i$ from the
stack. Then $q_i p_i$ tests if symbol $i$ is on the stack top, while $q_0 p_0$ tests
if the stack boundary $0$ is on top, so that the equation $q_0p_0+\ldots
+q_{m-1}p_{m-1}=1$ distinguishing $C_m$ from $C_m'$ asserts a \emph{completeness
  condition} for a stack with stack alphabet $\{1,\ldots,m-1\}$. For regular programs
$r \in \textit{RegExp\/}(\{q_0p_0,p_1,\ldots,q_{m-1}\})$, the scope $p_0 \ldots q_0$ of $p_0rq_0$
asserts that we start and end with an empty stack. Section \ref{sec-completeness}
shows that the completeness equation of $C_m$ in a sense already holds in $C_m'$ in
the scope of $p_0 \ldots q_0$.

Finally, the conclusion summarizes our results and indicates possible future extensions.

\section{\Star-continuous Kleene algebras and $\CR$-dioids}\label{sec-KA}

A \blue{Kleene algebra}, as defined in \cite{Kozen94}, is an idempotent semiring or dioid $(K,+,\cdot,0,1)$ with a unary operation ${}^* : K \to K$ such that for all $a, b \in K$
\begin{eqnarray*}
	a\cdot a^*+1 \leq a^*	&\wedge&	\forall x(a\cdot x + b \leq x \to a^*\cdot b \leq x), \\
	a^*\cdot a +1 \leq a^*	&\wedge&	\forall x(x\cdot a + b \leq x \to b\cdot a^* \leq x),
\end{eqnarray*}
where $\leq$ is the natural partial order on $K$ given by $a \leq b$ iff $a+b=b$.

A Kleene algebra is \emph{non-trivial} if $0\not=1$, and it \blue{has zero-divisors} if
there are non-zero elements $a,b$ such that $a\cdot b=0$.
The boolean Kleene algebra $\B = (\{0,1\},+,\cdot,{}^*,0,1)$ with boolean addition and
multiplication and ${}^*$ given by $0^*=1^*=1$ is a subalgebra of any
non-trivial Kleene algebra $K$.

A Kleene algebra $K = (K,+,\cdot,{}^*,0,1)$ is \blue{\Star-continuous} if
\[
	a \cdot c^*\cdot b = \Sum{a\cdot c^n\cdot b}{n\in\N}
\]
for all $a, b, c \in K$, where $\sum$ is the least upper bound with respect to the natural partial order.
Well-known \Star-continuous Kleene algebras are the algebras $\CR M = (\CR M, +,\cdot,{}^*,0,1)$
of regular subsets of monoids $M= (M,\cdot^M,1^M)$, where $0 := \emptyset$, $1 := \{1^M\}$ and for $A, B \in \CR M$,
\[ \renewcommand\arraystretch{1.4} \begin{array}{c@{\qquad}l}
	A+B = A \cup B,		        & A \cdot B = \setof{a\cdot^M b}{a\in A,b\in B},
\\	A^* = \Union{A^n}{n\in\N}	& \textrm{with } A^0 = 1, A^{n+1} = A \cdot A^n.
\end{array}
\]
If $K$ is a dioid $(K,+^K,\cdot^K,0^K,1^K)$ or a Kleene algebra, by \blue{$\CR K$} we mean the Kleene algebra $\CR M$ of its multiplicative monoid $M=(K,\cdot^K,1^K)$.

An \blue{$\CR$-dioid} is a dioid $K=(K,+^K,\cdot^K,0^K,1^K)$ where each $A \in \CR K$
has a least upper bound $\sum A \in K$, i.e.~$\sum$ is \blue{$\CR$-complete}, and
where $\sum(AB) = (\sum A)(\sum B)$ for all $A, B \in \CR K$, i.e.~$\sum$ is \blue{$\CR$-distributive}.
An \blue{$\CR$-morphism} is a dioid morphism that preserves least upper bounds of
regular sets.

Any $\CR$-dioid $K$ can be expanded to a \Star-continuous Kleene algebra by putting $c^* := \sum\{c\}^*$ for $c \in K$.
Conversely, the dioid reduct of a \Star-continuous Kleene algebra $K$ is an $\CR$-dioid,
since, by induction, every regular set $C$ has a least upper bound $\sum C \in K$ satisfying $a\cdot (\sum C)\cdot b = \sum (aCb)$,
which implies the $\CR$-distributivity property $\sum(AB) = (\sum A)(\sum B)$ for $A, B \in \CR K$ (see \cite{Hopkins-I-2008}).

The \Star-continuous Kleene algebras, with Kleene algebra homomorphisms (semiring
homomorphisms that preserve \Star), form a category. It is isomorphic to the category
\blue{$\D\CR$} of $\CR$-dioids and $\CR$-morphisms,
cf.~\cite{Kozen-Algorithms91,Hopkins-I-2008,LeissHopkins18a}, and a subcategory of the category $\D$ of
dioids and dioid morphisms.
There is an adjunction $(\CR,\widehat\CR,\eta,\epsilon)$
between the category $\M$ of monoids and the category $\D\CR$, where $\widehat\CR$ is
the forgetful functor, the unit $\eta$ is given by $\eta_M : M\to \CR M$ with
$\eta_M(m)=\{m\}$ and the counit $\epsilon$ by $\epsilon_K:\CR K\to K$ with
$\epsilon_K(A)=\sum A$, for monoids $M$ and $\CR$-dioids $K$, cf.~Theorem 16 of \cite{Hopkins-II-2008}.

The $\CR$-dioids of the form $\CR M$ with monoid $M$ form the Kleisli subcategory of $\D\CR$.
The cases of most immediate interest are the algebras $\CR X^*$ associated with
regular expressions and regular languages over an alphabet $X$, and $\CR\left(X^*
\times Y^*\right)$ of rational relations and rational transductions with alphabets
$X$ and $Y$, respectively, of inputs and outputs.

\subsection{The polycyclic $\CR$-dioids}\label{sec-polycyclic}

We will make use of two kinds of $\CR$-dioids which do not belong to the
Kleisli subcategory, but are quotients of the regular sets $\CR\Delta^*$ by
suitable $\CR$-congruence relations $\rho$ on $\CR\Delta^*$, where $\Delta$
is an alphabet of ``bracket'' pairs.
In this section, we introduce the polycyclic $\CR$-dioids $C_m'$ over an alphabet
$\Delta_m$ of $m$ bracket pairs; the bra-ket $\CR$-dioids $C_m$ over $\Delta_m$ are
deferred to Section \ref{sec-completeness}.

Let $\rho$ be a dioid congruence on an $\CR$-dioid $D$.
The set $D/\rho$ of congruence classes is a dioid under the operations defined by
$(d/\rho)(d'/\rho) := (dd')/\rho$, $1 := 1/\rho$, $d/\rho + d'/\rho := (d+d')/\rho$, $0 := 0/\rho$.
Let $\leq$ be the partial order on $D/\rho$ derived from $+$.
For $U \subseteq D$, put $U/\rho :=\setof{d/\rho}{d\in U}$ and
\[
	(U/\rho)^\da = \setof{e/\rho}{e/\rho \leq d/\rho \textrm{ for some }d \in U, e \in D}.
\]
An \blue{$\CR$-congruence} on $D$ is a dioid-congruence $\rho$ on $D$ such that for
all $U, U' \in \CR D$, if $(U/\rho)^\da = (U'/\rho)^\da$, then $(\sum U)/\rho = (\sum
U')/\rho$.
 It is easy to see that the kernel of an $\CR$-morphism is an $\CR$-congruence.

\begin{proposition}[Proposition 1 of \cite{LeissHopkins18a}]
	If $D$ is an $\CR$-dioid and $\rho$ an $\CR$-congruence on $D$, then $D/\rho$ is an $\CR$-dioid.
	For every $R \subseteq D\times D$ there is a least $\CR$-congruence $\rho \supseteq R$ on $D$.
\end{proposition}

Let $\Delta_m= P_m \dotcup Q_m$ be a set of $m$ ``opening brackets''
$P_m=\setof{p_i}{0\leq i < m}$ and $m$ ``closing brackets''
$Q_m=\setof{q_i}{0\leq i < m}$, with $P_m \cap Q_m=\emptyset$.
The \blue{polycyclic $\CR$-dioid $C_m'$} is the quotient $C_m' = \CR\Delta_m^*/\rho$ of $\CR\Delta_m^*$ by the $\CR$-congruence $\rho$ generated by the relations
\begin{equation} 
	\setof{p_iq_j = \delta_{i,j}}{i,j<m}.
\end{equation}
These equations allow us to algebraically distinguish matching brackets, where $p_iq_j=1$, from non-matching ones, where $p_iq_j=0$.
The \blue{polycyclic monoid $P_m'$} of $m$ generators is the quotient of $(\Delta_m \dotcup \{0\})^*$ by the monoid congruence $\sigma_m$ generated by
\begin{equation}\label{poly-Pn'}
	\setof{p_iq_j = \delta_{i,j}}{i,j<m} \cup \setof{x0 = 0}{x\in \Delta_m \dotcup \{0\}} \cup \setof{0x = 0}{x\in \Delta_m}. \nonumber
\end{equation}
Each element $w \in (\Delta_m \dotcup \{0\})^*$ has a \blue{normal form} $\nf w \in Q_m^*P_m^* \cup \{0\}$,
ob\-tained by using the equations to shorten $w$, that represents $w/\sigma_m \in P_m'$.
Hence,
\[
	P_m' \simeq (Q_m^*P_m^* \cup \{0\},\cdot,1) \quad\textrm{with }v\cdot w = \nf{vw}.
\]
The polycyclic $\CR$-dioid $C_m'$ can be understood as the regular sets of strings in
normal form:
\begin{proposition}[Proposition 9 of \cite{Leiss22}]\label{prop-C2'}
  Let $\nu$ be the least $\CR$-congruence on $\CR P_m'$ that identifies
  $\{0\}$ with the empty set. Then $C_m' \simeq \CR P_m'/\nu$ via the
  mapping defined by $A/\rho\mapsto \setof{\nf{w}}{w\in A}/\nu$ for
  $A\in\CR\Delta_m^*$. Each element $A/\rho$ of $C_m'$ is uniquely represented by a
  subset of $Q_m^*P_m^*$, namely $\setof{\nf{w}}{w\in A}\setminus \{0\}$.
\end{proposition}



The normal form can be extended from $P_m'$ to monoid extensions $P_m'[X]$ of $P_m'$ in which elements of $X$ are required to commute with elements of $P_m'$.
Formally, let $Y = \Delta_m \dotcup\{0\} \dotcup X$ and $P_m'[X]$ the quotient of $Y^*$ under the congruence generated by
(i) the matching rules $\setof{p_iq_j=\delta_{i,j}}{i,j<m}$,
(ii) the annihilation rules $y0 = 0$ and $0y=0$ for $y \in Y$, and
(iii) the commutation rules $\setof{xd = dx}{x\in X, d\in \Delta_m}$.
The set $Y^*$ can be decomposed into strings containing a 0, strings containing an opening bracket followed by
a symbol of $X$ or by a closing bracket, strings containing a symbol of $X$ followed by a closing
bracket, and strings consisting only of closing brackets followed by symbols of $X$ followed by
opening brackets, i.e.~
\[
	Y^* =  Y^*\{0\}Y^* \cup Y^*(P_mX \cup P_mQ_m \cup XQ_m)Y^* \cup Q_m^*X^*P_m^*.
\]
A normal form $\nf{w} \in Q_m^*X^*P_m^* \cup \{0\}$ for strings $w \in Y^*$ can hence
be obtained: use the annihilation rules to replace $u0v$ by 0, use the commutation rules to move opening brackets $p_i \in P_m$
to the right and closing brackets $q_i \in Q_m$ to the left of elements of $X^*$,
then use the matching rules to shorten $up_iq_jv$ to $uv$ or $u0v$, and repeat this process.
I.e.~for $i, j < m, i\not=j$ and $x \in X$, $u, v \in Y^*$ we put
\[
\begin{array}{rcl@{\qquad}rcl@{\qquad}rcl}
	\nf{up_ixv}	&:=& \nf{uxp_iv},	& \nf{u0v}	&:=& 0,	& \nf{up_iq_iv}	&:=& \nf{uv}, \\
	\nf{uxq_iv}	&:=& \nf{uq_ixv},	& \nf{1}	&:=& 1,	& \nf{up_iq_jv}	&:=& 0.
\end{array}
\]
We leave it to the readers to convince themselves that this amounts to a confluent rewriting system, so that $\tnf: Y^* \to Q_m^*X^*P_m^* \cup \{0\}$ is well-defined, and that
\begin{equation}\label{eqn-Pm'[X]}
	P_m'[X] \simeq (Q_m^*X^*P_m^* \cup \{0\},\cdot,1), \quad \textrm{where } u\cdot v := \nf{uv}.
\end{equation}

The normal form \tnf\ on $P_m'[X]$ is the motivating idea behind the normal form theorem (Theorem \ref{thm-nf-C'})
for elements of the tensor product $\CR X^*\xR C_m'$ to be introduced in the next section.
On the tensor product, regular sets $A \in \CR X^*$ and (congruence classes of) regular sets $B \in \CR\Delta_m$ commute with each other,
and the tensor product is an $\CR$-dioid structure, not just a monoid structure.

We notice that a suitable coding of $m \geq 2$ bracket pairs by two pairs extends to
an embedding of $C_m'$ in $C_2'$.
In the context $p_0\ldots q_0$, the code of any normal form $w\in Q_m^*P_m^*$ except
1 is annihilated.

\begin{lemma}\label{lem-recoding-C'}
	For $m \geq 2$ there is an embedding $\CR$-morphism $g: C_m' \to C_2'$
	such that for $i, j < m$,
        \[ g(p_i)\cdot g(q_j) = \delta_{i,j} \quad\textrm{and}\quad
        p_0\cdot g(q_j)=0 = g(p_i)\cdot q_0,
        \]
        where we wrote $p_i, q_j$ for the congruence class of $\{p_i\}, \{q_j\}$ in
        $C_m'$ and $C_2'$, respectively.
\end{lemma}

\eject

\begin{proof}
	Write $\Delta_m = P_m \dotcup Q_m$ with $P_m=\{p_0,\ldots,p_{m-1}\}$, $Q_m=\{q_0,\ldots,q_{m-1}\}$,
	but for $\Delta_2 = P_2 \dotcup Q_2$, use $b, p$ for $p_0, p_1$ and $d, q$ for $q_0, q_1$.
	Let $\ol{\,\cdot\,}: \Delta_m^* \to \Delta_2^*$ be the homomorphism generated
        by the coding of $\Delta_m$ in $\Delta_2^*$ by
	\[
		\ol{p_i}= bp^{i+1}, \quad \ol{q_i} = q^{i+1}d, \quad\textrm{for }i < m.
	\]
	The functor $\CR$ lifts $\ol{\,\cdot\,}$ by $\ol A := \setof{\ol w}{w\in A}$ to a monotone
        homomorphism $\ol{\,\cdot\,}:\CR\Delta_m^* \to \CR\Delta_2^*$;
        since the supremum $\sum$ on $\CR\CR\Delta_m^*$ and $\CR\CR\Delta_2^*$ is the union of sets, $\ol{\,\cdot\,}$ is an $\CR$-morphism.
	Let $\rho_m$ be the $\CR$-congruence on $\CR\Delta_m^*$ generated by the (semiring) equations
	\[
		p_iq_i=1, \quad p_iq_j=0, \quad \textrm{for }i\not=j < m.
	\]
	Then clearly
	\[
		\ol{p_i}\ol{q_j}/\rho_2 = bp^{i+1}q^{j+1}d/\rho_2 = \delta_{i,j} = p_iq_j/\rho_m
	\]
	and
	\[
		(b\cdot \ol{q_j})/\rho_2 = bq^{j+1}d/\rho_2 = 0 = bp^{i+1}d/\rho_2 = (\ol{p_i}\cdot d)/\rho_2.
	\]
	Extend the $\CR$-morphism $\ol{\,\cdot\,}:\CR\Delta_m^* \to \CR\Delta_2^*$ to a map $g: C_m' \to C_2'$ by
	\[
		g(A/\rho_m) := \ol{A}/\rho_2 \quad\textrm{for }A \in \CR\Delta_m^*.
	\]
	This map is well-defined and injective:
	by Proposition \ref{prop-C2'}, $A/\rho_m$ is represented by a set of strings in normal form, $\setof{\nf{w}}{w \in A} \setminus \{0\} \subseteq Q_m^*P_m^*$,
	and $\ol{\,\cdot\,}$ maps $Q_m^*P_m^*$ injectively to a set of normal form strings of $Q_2^*P_2^*$.

	Clearly, $g: C_m' \to C_2'$ is a monotone semiring morphism.
	Since $\cdot/\rho_m: \CR\Delta_m^* \to C_m'$ is surjective, $g: C_m' \to C_2'$ is an $\CR$-morphism:
	for each $U \in \CR C_m'$ there is $V \in \CR\Delta_m^*$ such that $U = \setof{A/\rho_m}{A \in V}$, hence
	\begin{eqnarray*}
		g(\sum U)	&=& g((\bigcup V)/\rho_m) = g(\bigcup V)/\rho_2
	\\			&=& (\Union{g(A)}{A\in V})/\rho_2
	\\			&=& \Sum{g(A)/\rho_2}{A\in V}
	\\			&=& \Sum{g(A/\rho_m)}{A\in V}
	\\ \hspace{4.2cm}	&=& \Sum{g(B)}{B\in U}. \hspace{4cm}
	\end{eqnarray*}
\end{proof}
Based on Lemma \ref{lem-recoding-C'}, in the following we state most results only for $m=2$.

\subsection{The tensor product $K\xR C$ of $\CR$-dioids $K$ and $C$}\label{sec-tensorproduct}

Two maps $\embeddings{f}{M_1}{M}{M_2}{g}$ to a monoid $M$ are \blue{relatively commuting} if $f(m_1)g(m_2) = g(m_2)f(m_1)$ for all $m_1\in M_1$ and $m_2\in M_2$.
In a category whose objects have a monoid structure, a \blue{tensor product} of two objects $M_1$ and $M_2$ is an object $M_1 \otimes M_2$
with two relatively commuting morphisms $\tensorproduct M\top$ such that for any pair $\embeddings f{M_1}M{M_2}g$ of relatively commuting morphisms
there is a unique morphism $h_{f,g}: M_1 \otimes M_2 \to M$ with $f=h_{f,g} \circ \top_1$ and $g=h_{f,g} \circ \top_2$.
That is, the diagram
\[
\begin{tikzcd}[row sep=1.7cm, column sep=1.2cm]
  M_1 \arrow[rd,"f\,\ " below] \arrow[r, "\top_1"]
  & M_1 \otimes M_2 \arrow[d,dotted, "h_{f,g}"]
  & \arrow[ld,"\!g"] \arrow[l,"\top_2" above] M_2
 \\        &       M         &
\end{tikzcd}
\]
can be uniquely completed as shown.
Intuitively, the tensor product $M_1 \otimes M_2$ is the free extension of $M_1$ and $M_2$ in which elements of $M_1$ commute with those of $M_2$.

In the category of monoids, $M_1 \otimes M_2$ is the cartesian product $M_1\times M_2$ with componentwise unit and product, and $h_{f,g}(m_1, m_2) = f(m_1) \cdot g(m_2)$.
The category $\D\CR$ of \Star-continuous Kleene algebras also has tensor products:
\begin{theorem}[Theorem 4 of \cite{LeissHopkins18a}]\label{thm-tensorproducts}
	Let $K_1, K_2$ be $\CR$-dioids and $M_1, M_2$ their multiplicative monoids.
	The tensor product of $K_1, K_2$ is
	\[
		K_1\xR K_2 := \CR(M_1\times M_2)/_\equiv,
	\]
	the quotient of the regular sets $\CR(M_1\times M_2)$ of the monoid product $M_1\times M_2$ by the $\CR$-congruence $\equiv$ generated by the ``tensor product equations''
	\[
		\setof{A\times B = \{(\sum A, \sum B)\}}{A\in \CR M_1, B\in\CR M_2}.
	\]
Since the natural embeddings of $M_1,M_2$ in $M_1\times M_2$
 lift $A \in \CR M_1$ and $B\in\CR M_2$ to sets in $\CR(M_1\times M_2)$,
\[
	A\times B = (A\times\{1\})(\{1\}\times B) \in \CR(M_1\times M_2).
\]
	The $\CR$-morphisms $\product K{\top}\xR$ are $\top_1(a) := \{(a, 1)\}/_\equiv$ for $a \in K_1$ and $\top_2(b) = \{(1, b)\}/_\equiv$ for $b \in K_2$.
	For a pair of commuting $\CR$-morphisms $\embeddings f{K_1}K{K_2}g$ to an $\CR$-dioid $K$, the induced map is
	\[
		h_{f,g}(R/_\equiv) := \Sum{f(a)g(b)}{(a,b)\in R}, \quad R \in \CR(M_1\times M_2).
	\]
\end{theorem}

For $a\in K_1$ and $b\in K_2$, the tensor $\top_1(a)\top_2(b) = \{(a,b)\}/_\equiv$\,
is written $a\otimes b$, but when $K_1$ and $K_2$ are disjoint, we simply use $ab$.
(If they are not disjoint, $ab$ could also mean $(ab\otimes 1)$ or $(1\otimes ab)$.)
Notice that if $a=0$ in $K_1$ or $b=0$ in $K_2$, then $a\otimes b = 0$ in $K_1\xR K_2$,
for if, say, $a=0$, then
\[\{(0,b)\}=\{(\sum\emptyset,\sum\{b\})\} \equiv \emptyset\times\{b\} = \emptyset.\]
It follows that $K_1\xR K_2$ is trivial if $K_1$ or $K_2$ is trivial.

\begin{proposition}[Proposition 7 of \cite{LeissHopkins18a}]\label{prop-R(MxN)}
If $M_1$ and $M_2$ are monoids, then $\CR M_1 \xR \CR M_2 \simeq \CR(M_1\times M_2)$.
\end{proposition}

\begin{shortproof}
Let $\top_1(A)= A\times \{1\}$ for $A\in \CR
M_1$ and $\top_2(B)= \{1\}\times B$ for $B\in \CR M_2$ in
\[
\begin{tikzcd}[row sep=1.8cm, column sep=1cm]
    \CR{M_1} \arrow[rd,"f\,\ " below] \arrow[r,"\top_1"]
    & \CR{(M_1\times M_2)} \arrow[d,dotted,"h_{f,g}"]
    & \arrow[ld,"\!g"]\arrow[l,"\top_2" above] \CR{M_2}
\\                 & K &
\end{tikzcd}
\]
and put $h_{f,g}(S) = \Sum{f(\{a\})g(\{b\})}{(a,b)\in S}$ for $S\in \CR(M_1\times M_2)$
and commuting $\CR$-morphisms $f,g$ to an $\CR$-dioid $K$. These satisfy the
properties of a tensor product of $\CR M_1$ and $\CR M_2$, so the claim holds by the
uniqueness of tensor products.
\end{shortproof}
In the following, for $\CR$-dioids $K_1, K_2$, we also write $K_1\times K_2$ for the
product of their underlying multiplicative monoids,
and for $R \in \CR(K_1\times K_2)$, we write \blue{$[R]$} instead of $R/_\equiv$.
For $R, S \in \CR(K_1\times K_2)$, one has $[R]+[S] = [R \cup S]$, $[R][S] = [RS]$, and
\[
  [R]^*	= \Sum{[R]^n}{n\in\N} = \Sum{[R^n]}{n\in\N} = [\,\Union{R^n}{n\in\N}\,] = [R^*].
\]
Notice also that $[R]=[\bigcup\setof{\{(a,b)\}}{(a,b)\in R}] = \Sum{a\otimes b}{(a,b)\in R}$.

The $\CR$-morphisms in $\product K{\top}\xR$ are embeddings, unless one of
$K_1$ and $K_2$ is trivial and the other is not:

\begin{lemma}\label{lem-non-triv}
  Let $K_1$ and $K_2$ be non-trivial $\CR$-dioids. Then the tensor product
  \[ \product K\top{\xR} \]
  is non-trivial, and $\top_1$ and $\top_2$ are embeddings.
\end{lemma}

\begin{proof} 
An element $(x,y)\in K_1\times K_2$ is an \blue{upper bound of} $R\subseteq K_1\times
K_2$, written $R\preceq (x,y)$, if $a\leq x$ and $b\leq y$ for all $(a,b)\in R$. Let
$Z = \setof{(a,b)\in K_1\times K_2}{a=0 \textrm{ or }b=0}$. For $R,S\in \CR(K_1\times
K_2)$, define $P(R,S)$ by
  \begin{equation}\label{eqn-Q}
    \forall (x,y),(a,b),(a',b')[ (a,b)R(a',b')\setminus Z \preceq (x,y)
      \iff  (a,b)S(a',b')\setminus Z \preceq (x,y)].
  \end{equation}
  With $(a,b)=(a',b')=(1,1)$, (\ref{eqn-Q}) says that $R\setminus Z$ and $S\setminus
  Z$ have the same upper bounds in $K_1\times K_2$.  In particular, for $R\subseteq
  Z$ and $S\not\subseteq Z$, $P(R,S)$ is false, since $(0,0)$ is an upper bound of
  $R\setminus Z$, but not of $S\setminus Z$.  We defer the proof of
  $\equiv\,\subseteq P$ to the appendix, Section \ref{sec-appendix}. Then for any
  $(a,b)\notin Z$, $\{(0,0)\}\not\equiv \{(a,b)\}$, hence $0\otimes 0 \not=a\otimes
  b$. As $(1,1)\notin Z$, $0 = 0\otimes 0 \not= 1\otimes 1 = 1$, so $K_1\xR K_2$ is
  non-trivial.  Furthermore, if $(a,b)$ and $(a',b')$ are different elements of
  $(K_1\times K_2)\setminus Z$, then $a\otimes b \not= a'\otimes b'$, because
  $\{(a,b)\}\equiv\{(a',b')\}$ implies, via (\ref{eqn-Q}), that $(a,b)$ is an upper
  bound of $\{(a',b')\}$ and $(a',b')$ is an upper bound of $\{(a,b)\}$, so
  $(a,b)=(a',b')$.  In particular, $\top_1$ and $\top_2$ are injective.
\end{proof}

\begin{corollary}\label{cor-tensor-ordering}
  If $K_1\xR K_2$ is the tensor product of non-trivial $\CR$-dioids
  $K_1$ and $K_2$, then for all $a,a'\in K_1$ and $b,b'\in K_2$,
  \begin{enumerate}
  \item if $a\otimes b = 0$, then $a = 0$ in $K_1$ or $b = 0$ in $K_2$,
  \item if $0\not=a\otimes b\leq a'\otimes b'$ in $K_1\xR K_2$, then $0\not=a\leq a'$ in
    $K_1$ and $0\not=b\leq b'$ in $K_2$.
  \end{enumerate}
\end{corollary}

\begin{proof}
  For (i), if $a\not=0$ and $b\not=0$, then $\{(0,0)\}\not\equiv \{(a,b)\}$ by the
  previous proof, which just means that $a\otimes b \not=0$ in $K_1\xR K_2$.  For
  (ii), suppose $0\not= a\otimes b \leq a'\otimes b'$. Then $0\not=a'\otimes b'$ too,
  and $a\not=0\not=a'$ in $K_1$ and $b\not=0\not=b'$ in $K_2$. Since $a\otimes b +
  a'\otimes b' = a' \otimes b'$, we have
  \[ \{(a,b),(a',b')\} \equiv \{(a',b')\}, \]
  and since $\equiv\,\subseteq P$ for the predicate $P$ in the proof of Lemma
  \ref{lem-non-triv}, for any $(x,y)$ we have, by (\ref{eqn-Q}), 
  \[ \{(a,b),(a',b')\} \preceq (x,y) \Leftrightarrow \{(a',b')\}\preceq (x,y). \]
  The right-hand side is true for $(x,y)=(a',b')$, so $(a,b)\leq (a',b')$ holds by the left-hand side.
\end{proof}

\begin{corollary}\label{cor-injectivity}
Let $f$ and $g$ be injective $\CR$-morphisms between non-trivial $\CR$-dioids in
\[
\begin{tikzcd}[row sep=1.6cm, column sep=1.1cm]
  K_1 \arrow[r,"\top_1"] \arrow[d,"f"]
  & K_1\xR K_2 \arrow[d,dotted,"{h_{\fg}}"]
  & \arrow[l,"\top_2" above]\arrow[d,"g"] K_2
\\    K_1' \arrow[r,"\top_1'"] & K_1'\xR K_2' & \arrow[l,"\top_2'" above] K_2'.
\end{tikzcd}
\]
Then $h_{\fg} : K_1\xR K_2\to K'_1\xR K'_2$, the induced
$\CR$-morphism for $\top_1'\circ f$ and $\top_2'\circ g$, is injective.
\end{corollary}

\begin{proof}
By Lemma \ref{lem-non-triv}, the $\CR$-morphisms $\top_1,\top_2,\top_1',\top_2'$ are
embeddings. The homomorphism $f\times g : K_1\times K_2\to K'_1\times K'_2$ lifts to
a monotone homomorphism $\fg:\CR(K_1\times K_2)\to \CR(K'_1\times K'_2)$.  Since
$\top_1'\circ f$ and $\top_2'\circ g$ are commuting, they induce an $\CR$-morphism $h
= h_{\fg}$. For $R\in\CR(K_1\times K_2)$, it maps $[R]\in K_1\xR K_2$ to
\[ h([R]) = [\fg(R)]' = \sum{'}\setof{fa\otimes' gb}{(a,b)\in R}
\in K'_1\xR K'_2,
\]
where $fa\otimes' gb = \top_1'(fa)\top_2'(gb)$ and $\sum{'}$ is the least upper bound
of the $\CR$-dioid $K'_1\xR K'_2$. In particular, for $(a,b)\in K_1\times K_2$,
\[ h(a\otimes b) = fa\otimes' gb. \]
By Lemma \ref{lem-non-triv}, $h$ is monotone and injective on the image of
$K_1\times K_2$ under $\otimes$.
To see that $h$ is injective, suppose $R,S\in \CR(K_1\times K_2)$ and
\[ [R] = \sum\setof{a\otimes b}{(a,b)\in R} \not=
         \sum\setof{a\otimes b}{(a,b)\in S} = [S].
\]
Then $\setof{a\otimes b}{(a,b)\in R}^\da\not= \setof{a\otimes b}{(a,b)\in S}^\da$ by
the definition of $\sum$ of $K_1\xR K_2$.  Since $h$ is
monotone and injective on the image of $K_1\times K_2$ under $\otimes$,
\[ \setof{fa\otimes' gb}{(a,b)\in R}^\da \not= \setof{fa\otimes' gb}{(a,b)\in S}^\da. \]
Then we must have 
\[ h([R]) = \sum{'}\setof{fa\otimes' gb}{(a,b)\in R}^\da \not=
            \sum{'}\setof{fa\otimes' gb}{(a,b)\in S}^\da = h([S]),
\]
as otherwise $\fg(R)\equiv'\fg(S)$ for the $\CR$-congruence $\equiv'$ on
$\CR(K'_1\times K'_2)$ defining $K'_1\xR K'_2$, and $\equiv'$ were not the \emph{least}
$\CR$-congruence on $\CR(K'_1\times K'_2)$ containing the tensor product equations.
\end{proof}

We will mainly consider tensor products $K\xR C$ where $K=\CR X^*$ and $C$ is a
polycyclic $\CR$-dioid $C_m'$ or bra-ket $\CR$-dioid $C_m$.
For $L\in \CR X^*$, we have $\setof{\{w\}}{w\in L}\in\CR(\CR X^*)$, and since $\top_1$ is an $\CR$-morphism,
  \[ L\otimes 1 = \top_1(\bigcup\setof{\{w\}}{w\in L}) = \sum\setof{\{w\}\otimes 1}{w\in L}
  \in \CR X^*\xR C_2'.\]
The interest in Kleene algebras $\CR X^*\xR C_2'$ comes from the fact that $\CC X^*$,
the set of context-free languages over $X$, embeds in ${\CR X^*}\xR {C_2'}$, via
\[ L \in \CC X^* \mapsto \sum L := \Sum{\{w\}\otimes 1}{w\in L} \in \CR X^*\xR C_2', \]
cf.~Theorem 17 of \cite{Leiss22}. Notice that $L\otimes 1$ need not exist for
non-regular $L$.
Since all elements of $\CR X^*\xR C_2'$ can be denoted by regular expressions over $X
\dotcup \Delta_2$, every \emph{context-free} set $L \subseteq X^*$ is represented by
the value of a \emph{regular} expression.

\begin{example}\label{expl-1}
	Suppose $a, b \in X$.
	Then $L=\setof{a^nb^n}{n\in\N} \in \CC X^*$ is represented in $\CR X^*\xR
        C_2'$ by the value of the regular expression $r_L := \p0 (a\p1)^*(\q1b)^*
        \q0$ over $X \dotcup \Delta_2$. Writing elements of $X$ and $\Delta_2$ for
        their values in $\CR X^*\xR C_2'$, we have
        \[
	\renewcommand{\arraystretch}{1.3}
        \begin{array}{rcl@{\qquad}l}\qquad
          r_L	&=&	\Sum{\p0(a\p1)^n(\q1b)^m\q0}{n,m\in\N}	& \textrm{(\Star-continuity)} \\
	  &=&	\Sum{a^n\p0\p1^n\q1^m\q0b^m}{n,m\in\N}	& \textrm{(relative commutativity)} \\
	  &=&	\Sum{a^nb^n}{n\in\N}			& (\textrm{bracket match } \p{i}\q{j} = \delta_{i,j}).\qquad\triangleleft
	\end{array}
	\]
\end{example}

In the cases $K\xR C_m'$ of our main interest, where $K=\CR X^*$ and the polycyclic
$\CR$-dioid $C_m'\simeq \CR P_m'/\nu$ is a suitable quotient of $\CR P_m'$, the
tensor product construction can be replaced by a quotient construction. This is a
consequence of the following extension of Proposition \ref{prop-R(MxN)}.

\begin{theorem}\label{thm-CA(MxN)/tnu}
Let $M$ be a monoid and $N$ a monoid with annihilating element $0$. Then
\[ \CR M\xR (\CR N/\nu) \simeq \CR(M\times N)/{\tilde\nu}, \]
where $\nu$ is the least $\CR$-congruence on $\CR N$ containing
$(\{0\},\emptyset)$ and ${\tilde\nu}$ is the least $\CR$-congruence on
$\CR(M\times N)$ containing $(\{(1,0)\},\emptyset)$.

Putting $R_\nu := \setof{(A,B/\nu)}{(A,B)\in R}$ for $R\in \CR(\CR
M\times \CR N)$, the isomorphism is given by 
\begin{eqnarray*}
    [R_\nu] &\mapsto& (S_R)/{\tilde\nu}, \quad\textrm{where } S_R :=
    \bigcup\setof{A\times B}{(A,B)\in R} \textrm{ for }R\in
    \CR(\CR M\times \CR N), 
\\ S/{\tilde\nu} &\mapsto& [(R_S)_\nu], 
\quad\textrm{where }R_S := \setof{(\{m\},\{n\})}{(m,n)\in S} \textrm{ for }S\in \CR(M\times N).
\end{eqnarray*}
\end{theorem}
\begin{shortproof}This is an instance of Theorem 12 of \cite{Leiss22}.
\end{shortproof}

For $A\in\CR M$ and $B\in \CR N$, the isomorphism maps $A\otimes B/\nu$ to $(A\times
B)/\tilde\nu$, where $B/\nu$ is uniquely represented by $B\setminus\{0\}$ and
$(A\times B)/\tilde\nu$ by $(A\times B)\setminus (A\times\{0\})$. As $C_m'\simeq
P_m'/\nu$ by Proposition \ref{prop-C2'}, an application of the theorem is
\[ \CR X^*\xR C_m'\simeq \CR(X^*\times P_m')/{\tilde\nu}. \]
Moreover, since elements of $X$ and $P_m'$ commute in the monoid $P_m'[X]$ of
(\ref{eqn-Pm'[X]}),
\[ \CR(X^*\times P_m')/{\tilde\nu}\simeq \CR(P_m'[X])/\nu. \]
It follows that an element of $\CR X^*\xR C_m'$ has a unique representation by a subset
of $Q_m^*X^*P_m^*$.

However, to state our results for arbitrary $\CR$-dioids
$K$, we do need the tensor product $K\xR C_m'$.

\subsection{The centralizer $\zc{K}{C_2'}$ of $C_2'$ in $K\xR C_2'$}\label{sec-centralizer}

In a monoid $M$, the \blue{centralizer $Z_C(M)$} of a set $C \subseteq M$ in $M$ consists of those elements that commute with every element of $C$, i.e.~the submonoid
\[
	Z_C(M) :=\setof{m\in M}{mc = cm \textrm{ for all }c \in C}.
\]
For example, the centralizer of $\Delta_m$ in $P_m'[X]$ is $X^* \cup \{0\}$.

In Section \ref{sec-Normalforms}, we will, for non-trivial $\CR$-dioids $K$, consider
the representation of elements of $K\xR C_m'$ by automata.  As
$\top_1,\top_2$ in $\embeddings{\top_1}{K}{K\xR C_2'}{C_2'}{\top_2}$ are
relatively commuting, for all $k\in K$ and $c\in C_2'$ we have
\[ kc = \top_1(k)\cdot \top_2(c) = k\otimes c = \top_2(c)\cdot\top_1(k) = ck, \]
in $K\xR C_2'$, so $K\subseteq \zc{K}{C_2'}$ (modulo $\top_1$).
Moreover, $\zc{K}{C_2'}$ clearly is a semiring and, by \Star-continuity of $K\xR
C_2'$, it is closed under \Star: if $a$ commutes with $c\in C_2'$, then
\[ c\cdot a^* = \Sum{c\cdot a^n}{n\in\N} = \Sum{a^n\cdot c}{n\in\N} = a^*\cdot c.\]
In fact, $\zc{K}{C_2'}$ is an $\CR$-dioid, by Proposition 24 of \cite{Leiss22}.
It has even stronger closure properties, see Theorem \ref{thm-centralizer},
\ref{item-4} below.

A \blue{Chomsky algebra} (Grathwohl e.a.~\cite{Kozen2013}) is an idempotent semiring $D$ which is \blue{algebraically closed}, i.e.~every finite inequation system
\[
	x_1 \geq p_1(x_1, \ldots, x_k), \ldots, x_k \geq p_k(x_1, \ldots, x_k)
\]
with polynomials $p_1, \ldots, p_k \in D[x_1, \ldots, x_k]$ has a least solution in $D$, where $\leq$ is the partial order on $D$ defined by $a\leq b \iff a+b=b$.
Semiring terms over an infinite set $X$ of variables can be extended by a least-fixed-point operator $\mu$, such that if $t$ is a term and $x\in X$, $\mu x.t$ is a term.
In a Chomsky algebra $D$ with an assignment $h:X\to D$, the value of $\mu x.t$ is the least solution of $x\geq t$ with respect to $h$, i.e.~the least $a\in D$ such
that $x \geq t$ is true with respect to $h[x/a]$.
A Chomsky algebra $D$ is \blue{$\mu$-continuous}, if for all $a,b\in D$ and $\mu$-terms $t$,
\[ a\cdot \mu x.t \cdot b = \Sum{a\cdot t^n\cdot b}{n\in \N} \]
is true for all assignments $h:X\to D$, where $t^0 = 0$, $t^{n+1} = t[x/t^n]$.
The \Star-continuity condition of $\CR$-dioids is a special instance of the $\mu$-continuity condition, where $c^* = \mu x.(cx+1)$.
The semiring $\CC X^*$ of context-free languages over $X$ is a $\mu$-continuous Chomsky algebra.
The $\mu$-continuous Chomsky algebras, with fixed-point preserving semiring homomorphisms, form a category of dioids.

This category had been introduced as the category \blue{$\D\CC$} of $\CC$-dioids and $\CC$-morphisms in \cite{Hopkins-I-2008} as follows; for the equivalence, see \cite{LeissHopkins18b}.
For monoids $M$, let \blue{$\CC M$} be the semiring $(\CC M,\cup,\cdot,\emptyset,\{1\})$ of context-free subsets of $M$.
A \blue{$\CC$-dioid} $(M,\cdot,1,\leq,\sum)$ is a partially ordered monoid
$(M,\cdot,1,\leq)$ with an operation $\sum:\CC M\to M$ that is \blue{$\CC$-complete}
and \blue{$\CC$-distributive}, i.e.~
\begin{enumerate}
  \item[(i)] for each $A\in \CC M$, $\sum A$ is the least upper bound of $A$ in $M$ with respect to $\leq$,
 \item[(ii)] for all $A,B\in\CC M$, $\sum(AB) = (\sum A)\cdot(\sum B)$.
\end{enumerate}
A \blue{$\CC$-morphism} is a monotone homomorphism between $\CC$-dioids that preserves least upper bounds of context-free subsets.
The above mentioned strong closure property of the centralizer of $C_2'$ in $K\xR C_2'$ is that it is algebraically closed, which follows from \ref{item-4} of the following facts:

\begin{theorem}[Theorem 27, Lemma 30, Lemma 31 of \cite{Leiss22}]\label{thm-centralizer}
Let $M$ be a monoid and $K$ an $\CR$-dioid.
  \begin{enumerate}
  \item $\zc{K}{C_2'} = \setof{[R]}{R\in \CR(K\times C_2'), R\subseteq K\times \{0,1\}}$.
\label{item-3}
  \item $\zc{K}{C_2'}$ is a $\CC$-dioid.\label{item-4}
  \item The least-upper-bound operator $\sum:\CC K \to \zc{K}{C_2'}$ is a surjective homomorphism.\label{item-2}
  \item The least-upper-bound operator $\sum:\CC M \to \zc{\CR M}{C_2'}$ is a $\CC$-isomorphism.\label{item-1}
  \end{enumerate}
\end{theorem}
While \ref{item-3} gives a characterization of the elements in the centralizer of $C_2'$ in $K\xR C_2'$, in Section \ref{sec-Normalforms} we provide descriptions of \emph{all} elements of $K\xR C_2'$ via normal forms.
As the proof of Theorem \ref{thm-centralizer} is lengthy, we try to avoid using \ref{item-3} - \ref{item-1} as far as possible.
However, we need \ref{item-3} in the following corollary, which in turn is used to give a simplified normal form for elements of the centralizer in
Corollary \ref{cor-nf-C'}, and we use \ref{item-4} in Example \ref{expl-3ff} and for
the product case of Theorem \ref{thm-combining-nfs}.

A subset $X$ of a partial order $(P, \leq)$ is \blue{downward closed}, if for all $a, b \in P$, if $b \in X$ and $a \leq b$, then $a \in X$.

\begin{corollary}\label{cor-zc-bounded}
	If $K$ is a non-trivial $\CR$-dioid and has no zero divisors, then
        $\zc{K}{C_2'}$ is a downward-closed subset of $K\xR C_2'$.
\end{corollary}

\begin{proof}
  Suppose $[R]\leq [S]\in \zc K{C_2'}$ for $R,S\in\CR(K\times C_2')$. By Theorem
  \ref{thm-centralizer} \ref{item-3}, we can assume $S \subseteq K \times \{0, 1\}$
  and must show that there is $R'\in\CR(K\times C_2')$ with $[R]=[R']$ and
  $R'\subseteq K\times \{0,1\}$. The projection from $K\times C_2'$ to $K$ lifts to a
  homomorphism $\pi:\CR(K \times C_2') \to \CR K$, so $A :=\pi(S) \in \CR K$. Then
  \[
  S \subseteq A\times \{0, 1\} \in \CR(K\times C_2'),
  \]
  and for each $(k,c)\in R$,
  \[
  k\otimes c \leq [R] \leq [S] \leq [A\times\{0, 1\}] = [\{(\sum A,1)\}] = (\sum
  A)\otimes 1.
  \]
  If $0\not= k\otimes c$, then $c \leq 1$ in $C_2'$ by Corollary \ref{cor-tensor-ordering};
  by Proposition \ref{prop-C2'}, $c\in\{0,1\}$, so $(k,c)\in K\times\{0,1\}$.  If
  $0=k\otimes c$, then by Corollary \ref{cor-tensor-ordering} again, either $c=0$ and $(k,c)\in
  K\times \{0,1\}$, or else $k=0$.  Let $R' = R\setminus \setof{(0,c)\in R}{c\in
    C_2'}$. Then $R'\subseteq K\times\{0,1\}$ and $[R]=\Sum{k\otimes c}{(k,c)\in R}$ is the
  least upper bound of $\setof{k\otimes c}{(k,c)\in R'}$. We show by induction on the
  construction of $R\in\CR(K\times C_2')$ that $R'\in \CR(K\times C_2')$. This also
  gives $[R]=[R']$.

  If $R$ is finite, so is $R'$, therefore $R'\in \CR(K\times C_2')$. Suppose for
  $R_i\in\CR(K\times C_2')$, $i=1,2$, we have $R_i' = R_i\setminus\setof{(0,c)}{c\in
    C_2'}\in\CR(K\times C_2')$. If $R=R_1\cup R_2$, then $R' = R_1'\cup R_2'\in
  \CR(K\times C_2')$. If $R=R_1R_2$, then $R'\subseteq R_1'R_2'$, and since $K$ has
  no zero divisors, $R_1'R_2'\subseteq R'$, so $R'=R_1'R_2'\in \CR(K\times C_2')$. If
  $R=R_1^*$, then $R'=(\bigcup\setof{R_1^n}{n\in\N})' =
  \bigcup\setof{(R_1')^n}{n\in\N} = (R_1')^*\in \CR(K\times C_2')$.
\end{proof}

\subsection{Automata over a Kleene algebra}\label{sec-automata}

A \emph{finite automaton $\CA = \Z{S, A, F}$ with $n$ states} over a Kleene algebra $K$ consists of a transition matrix $A \in \mat nnK$
and two vectors $S \in \mat 1n\B$ and $F \in \mat n1\B$, coding the initial and final states.
The 1-step transitions from state $i < n$ to state $j < n$ are represented by $A_{i, j}$, and paths from $i$ to $j$ of finite length by $A^*_{i, j}$, where $A^*$ is the iteration of $A$.
The sum of paths leading from initial to final states defines an element of $K$,
\[
	L(\CA) = S\cdot A^*\cdot F\ \in K.
\]
The iteration $M^*$ of $M \in \mat nnK$ is defined by induction on $n$: for $n=1$ and $M=(k)$, $M^* = (k^*)$, and for $n > 1$,
\begin{equation}\label{mat-star}
	M^* = \begin{pmatrix}
		A & B \\
		C & D
	\end{pmatrix}^* = \begin{pmatrix}
		F^*		& F^* BD^* \\
		D^* CF^*\	& D^* CF^* BD^* + D^*
	\end{pmatrix},
\end{equation}
where $F = A+BD^*C$ and $M = \begin{pmatrix} A & B \\ C & D \end{pmatrix}$ is any splitting of $M$ in which $A$ and $D$ are square matrices of dimensions $n_1, n_2< n$ with $n=n_1+n_2$.

By Kleene's representation theorem, the set $\CR X^*$ of regular subsets of $X^*$ consists of the languages
\[
	L(\CA) = S\cdot A^*\cdot F\ \subseteq X^*
\]
of finite automata $\CA = \Z{S, A, F}$, where for some $n \in \N$, $A \in \mat nn{(\CF
  X^*)}$, $S \in \B^{1\times n}$, $F \in \mat n1\B$ and $\CF X^*$ is the set of finite subsets of $X^*$.

For various notions of Kleene algebra, Conway showed that the set $\mat nnK$ of $n\times
n$-matrices over $K$ with matrix addition, multiplication
and iteration as defined above and zero and unit matrices $0_n,1_n\in\mat nnK$ form a Kleene algebra
\[ \Mat nnK = (\mat nnK, +, \cdot, {}^*, 0_n,1_n) \]
and used this to prove
Kleene's representation theorem, see \cite{Conway71}. For the notion of Kleene
algebra used here, the same has been done by Kozen in \cite{Kozen94}. We are mostly
working with $\CR$-dioids, i.e.~\Star-continuous Kleene algebras, and will often make
use of \Star-continuity on the matrix level in Section \ref{sec-Normalforms}. In
fact, the $n\times n$-matrices over a \Star-continuous Kleene algebra form a
\Star-continuous Kleene algebra:

\begin{theorem}[Kozen \cite{Kozen-Algorithms91}, Chapter 7.1.]\label{thm-mat-starKA}
	If $K$ is a \Star-continuous Kleene algebra, so is $\Mat nnK$, for $n \geq 1$.
\end{theorem}

We remark that $\Mat nnK$ can be reduced to the tensor product of $K$ with $\Mat
nn\B$, but we will use this only in connection with bra-ket $\CR$-dioids in Section
\ref{sec-completeness}.

\begin{proposition}\label{prop-MatK} For any $\CR$-dioid $K$ and $n\geq 1$, $\Mat
  nnK \simeq K\xR \Mat nn\B$.
\end{proposition}

\begin{proofsketch}
One shows that  $\embeddings{I_K}{K}{\Mat nnK}{\Mat nn\B}{\Id}$ has the properties of
a tensor product, where $I_K(a) := a1_n$ for $a\in K$ and $\Id(B) = B$ for $B\in\mat nn\B$. 
For relatively commuting $\CR$-morphisms $f:K\to D\leftarrow \Mat nn\B:g$ to an
$\CR$-dioid $D$, the unique $\CR$-morphism with $f=h_{f,g}\circ I_K$ and
$g=h_{f,g}\circ\Id$ is defined by 
\[ h_{f,g}(A) := \sum\setof{f(A_{i,j})g(E_{(i,j)})}{i,j<n}, \quad
\textrm{for }A\in\mat nnK, \]
where $E_{(i,j)}\in\mat nn\B$ is the matrix with 1 only in
line $i$, row $j$. The claim then follows by the uniqueness of tensor products.
\end{proofsketch}

For any $\CR$-dioid $K$, we next prove Kleene's representation theorem for $K\xR
C_2'$: any element of $K\xR C_2'$ is the ``language'' $L(\CA) = SA^*F$ of a finite
automaton $\CA=\Z{S, A, F}$ over $K\xR C_2'$. This follows the proofs by Conway and
Kozen; the point here is how transitions by elements of $C_2'$ in the transition
matrix $A$ can be reduced to transitions by generators $c\in\Delta_2$ of $C_2'$.

For $a \in K$ and $c \in C_2'$, we write $a$ and $c$ also for their images in $K\xR
C_2'$, likewise $ac$ for their product in $K\xR C_2'$. 
From now on, for $\Delta_2=P_2\dotcup Q_2$ we use $P_2=\{b,p\}$ instead of
$\{\p0,\p1\}$ and $Q_2=\{d,q\}$ instead of $\{\q0,\q1\}$, unless stated otherwise.

\begin{theorem}\label{thm-aut}
Let $K$ be an $\CR$-dioid, i.e.~a \Star-continuous Kleene-algebra, and $C_2'$ the polycyclic Kleene algebra over $\Delta_2$.
For each $\phi \in K\xR C_2'$ there are $n \in \N$, $S \in \B^{1\times n}, F \in \B^{n\times 1}$, $U \in \{0, b, p\}^{n\times n}$, $V \in \{0, d, q\}^{n\times n}$ and $X \in K^{n\times n}$ such that
\[
	\phi = S(U+X+V)^*F.
\]
\end{theorem}

\begin{proof}
Since $\phi=[R]$ for some $R\in\CR(K\times C_2')$, by induction on the
construction of $R$ we build an automaton $\CA_R = \langle S,A,F\rangle$ over
$K\xR C_2'$ such that $L(\CA_R) = [R]$ and $A$ splits as $U+X+V$ as in the claim.

\begin{itemize}
\item $R=\emptyset$: Let $\CA_R = \langle S,A,F\rangle$ be the automaton of dimension 1
  with $S = (0)$, $A= (0)$, $F=(0)$. Then $L(\CA_R) = 0 = [\emptyset]$. We
  have $A= U+X+V$ with $1\times 1$ zero matrices $U,X,V$.

\item $R=\{(k,c)\}$ with $k\in K, c\in C_2'$: Since $\{(k,c)\} =
  \{(k,1)\}\cdot \{(1,c)\}$, by the product case below we may assume $k=1$
  or $c=1$. In the case $R=\{(k,1)\}$, let $\CA_R = \langle S,A,F\rangle$ consist of
\[ S = \begin{pmatrix} 1 & 0 \end{pmatrix}, \quad 
A = \begin{pmatrix} 1 & k \\ 0   & 1 \end{pmatrix}, \quad
F = \begin{pmatrix}0 \\ 1 \end{pmatrix}.
\]
Then $A^* = A$, since $A^0\leq A = A^2$, hence $L(\CA_R) = A_{1,2} = k1 = [\{(k,1)\}]$. The splitting is
\[ A = \begin{pmatrix} 1 & k \\ 0   & 1 \end{pmatrix}
     = \begin{pmatrix} 0 & 0 \\ 0   &  0 \end{pmatrix}
     + \begin{pmatrix} 1 & k \\ 0   & 1 \end{pmatrix}
     + \begin{pmatrix} 0 & 0 \\ 0   & 0 \end{pmatrix}
     = U+X+V.
\]

For the case $R=\{(1,c)\}$, the element $c\in C_2'$ is the congruence class of a set
$C\in\CR\Delta_2^*$ under the $\CR$-congruence $\rho_2'$ generated by the match
relations, so we can view $c$ as a regular expression in the letters of
$\Delta_2$. By the tensor product equations of $K\xR C_2'$, 
\[ \{(1,c_1+c_2)\} \equiv \{1\}\times \{c_1,c_2\} = \{(1,c_1)\}\cup\{(1,c_2)\}, \]
and since $\{(1,c_1c_2)\}=\{(1,c_1)\}\{(1,c_2)\}$ and $\{(1,c_1^*)\}=\{(1,c_1)\}^*$, we can construct $\CA_R$ by induction on the cases $R=R_1\cup R_2$, $R=R_1R_2$, and $R=R_1^*$ below.
In the remaining cases, $c$ is $0$, $1$ or a letter from $\Delta_2$. 
Let $\CA_R = \langle S,A,F\rangle$ consist of
\[ S = \begin{pmatrix} 1 & 0 \end{pmatrix}, \quad 
A = \begin{pmatrix} 1 & c \\ 0   & 1 \end{pmatrix}, \quad
F = \begin{pmatrix}0 \\ 1 \end{pmatrix}.
\]
Then $A^* = A$ and $L(\CA_R) = A_{1,1} = c = [\{(1,c)\}]$. 
If $c\in Q_2=\{d,q\}$, the splitting of $A$ is
\[ A = \begin{pmatrix} 1 & c \\ 0 & 1 \end{pmatrix}
     = \begin{pmatrix} 0 & 0 \\ 0 & 0 \end{pmatrix}
     + \begin{pmatrix} 1 & 0 \\ 0 & 1 \end{pmatrix}
     + \begin{pmatrix} 0 & c \\ 0 & 0 \end{pmatrix}
     = U+X+V.
\]
If $c\in P_2=\{b,p\}$, we switch the roles of $U$ and $V$. If $c$ is 0 or 1, let
\[ A = \begin{pmatrix} 1 & c \\ 0 & 1 \end{pmatrix}
     = \begin{pmatrix} 0 & 0 \\ 0 & 0 \end{pmatrix}
     + \begin{pmatrix} 1 & c \\ 0 & 1 \end{pmatrix}
     + \begin{pmatrix} 0 & 0 \\ 0 & 0 \end{pmatrix}
     = U+X+V.
\]

\item $R=R_1 \cup R_2$: For $i=1,2$, let $\CA_{R_i} = \langle S_i,A_i,F_i\rangle$ be an automaton of dimension $n_i$ such that
\[ L(\CA_{R_i})= S_iA_i^*F_i = [R_i] . \]
Construct $\CA_R = \langle S,A,F\rangle$ of dimension $n_1+n_2$ by
\[ S = \begin{pmatrix} S_1 & S_2 \end{pmatrix}, \quad 
A = \begin{pmatrix}A_1 & 0 \\ 0 & A_2 \end{pmatrix}, \quad
F = \begin{pmatrix}F_1 \\ F_2 \end{pmatrix}.
\]
By the recursion formula for iteration matrices,
\begin{eqnarray*}
  L(\CA_R) &=& SA^*F
  = \begin{pmatrix} S_1 & S_2 \end{pmatrix}
  \cdot\begin{pmatrix}A_1^* & 0 \\ 0 & A_2^*\end{pmatrix}
  \cdot \begin{pmatrix} F_1 \\ F_2 \end{pmatrix}
\\ &=& S_1A_1^*F_1 + S_2A_2^*F_2 
\\ &=& [R_1] + [R_2] = [R_1\cup R_2] = [R].
\end{eqnarray*}
The given splittings $A_1=U_1+X_1+V_1$ and $A_2=U_2+X_2+V_2$ combine to a suitable splitting of $A$ by
\[ A = \begin{pmatrix}U_1 & 0 \\ 0 & U_2 \end{pmatrix}
+ \begin{pmatrix}X_1 & 0 \\ 0 & X_2 \end{pmatrix}
+ \begin{pmatrix}V_1 & 0 \\ 0 & V_2 \end{pmatrix}
= U+X+V .\]

\item $R=R_1R_2$:
For $i=1,2$, let $\CA_{R_i} = \langle S_i,A_i,F_i\rangle$ be an automaton of dimension $n_i$ such that 
\[ L(\CA_{R_i})= S_iA_i^*F_i = [R_i]. \]
  
Construct $\CA_R = \langle S,A,F\rangle$ of dimension $n_1+n_2$ by
\[ S = \begin{pmatrix} S_1 & 0  \end{pmatrix}, \quad 
A = \begin{pmatrix}A_1 & F_1S_2 \\ 0   & A_2 \end{pmatrix}, \quad
F = \begin{pmatrix}0 \\ F_2 \end{pmatrix}.
\]
By the recursion formula for iteration matrices,
\begin{eqnarray*}
  L(\CA_R) &=& SA^*F 
\\ &=& \begin{pmatrix} S_1 & 0 \end{pmatrix}\cdot 
               \begin{pmatrix} A_1^* & A_1^*F_1S_2A_2^* 
                               \\ 0 & A_2^* \end{pmatrix}\cdot
               \begin{pmatrix} 0 \\ F_2 \end{pmatrix}
\\ &=& S_1A_1^*F_1S_2A_2^*F_2 
\\ &=& [R_1][R_2] = [R_1R_2] = [R].
\end{eqnarray*}
The given splittings $A_1=U_1+X_1+V_1$ and $A_2=U_2+X_2+V_2$ combine to the splitting
\[ A = \begin{pmatrix}U_1 & 0 \\ 0 & U_2 \end{pmatrix}
+ \begin{pmatrix}X_1 & F_1S_2 \\ 0 & X_2 \end{pmatrix}
+ \begin{pmatrix}V_1 & 0 \\ 0 & V_2 \end{pmatrix}
= U+X+V .\]

\item $R=R_1^*$: Suppose $\CA_{R_1} = (S_1,A_1,F_1)$, is an automaton such that 
\[ L(\CA_{R_1}) = S_1A_1^*F_1 = [R_1]. \]
Let $\CA_{R^+} = \langle S,A,F\rangle$ be $\langle
S_1,A_1+F_1S_1,F_1\rangle$. By equalities in Kleene algebras,
\begin{eqnarray*}
  L(\CA_{R^+}) &=& S_1(A_1+F_1S_1)^* F_1
\\ &=& S_1A_1^*(F_1S_1A_1^*)^*F_1
\\ &=& S_1A_1^*F_1(S_1A_1^*F_1)^*
\\ &=& [R_1][R_1]^*
\\ &=& [R_1][R_1^*] = [R_1^+],
\end{eqnarray*}
The splitting $A=U+X+V$ is obtained from the splitting $A_1=U_1+X_1+V_1$ by $U=U_1$, $X= X_1+F_1S_2$ and $V=V_1$. Finally, put $\CA_{R^*} = \CA_{\{(1,1)\}\cup R^+}$ and split its transition matrix as shown for the case $\CA_{R_1\cup R_2}$.
\QED
\end{itemize}
\def\QED{}
\end{proof}

\section{Normal form theorems for $K\xR C_2'$ with $\CR$-dioid $K$}
\label{sec-Normalforms}

In the representation of elements $\phi$ of $K\xR C_2'$ as $\phi=L(\CA)=SA^*F$ by automata $\CA = \Z{S, A, F}$ with $A = U+X+V$ in Theorem \ref{thm-aut}, $A^*=(U + X + V)^*$
admits arbitrary sequences of opening brackets $U$ with closing brackets $V$.
We aim at a normal form for $(U + X + V)^*$ where brackets are mainly occurring in a balanced way.
To this end, we now look at ways to express a Dyck-language with a single bracket pair $u, v$ in a Kleene algebra.

\subsection{Least solutions of some polynomial inequations in Kleene algebras}

We first show that in any Kleene algebra $K$, if they exist, least solutions of two fixed-point inequations that might be used to define Dyck's language $D_1(X)$ with $X=\{x_1, \ldots, x_n\} \subseteq K$, namely
\[
	y \geq (x_1+\ldots+x_n + uyv)^* \quad\textrm{and}\quad y \geq 1 + x_1+\ldots+x_n + uyv + yy,
\]
are related, where $u, v \in K \setminus X$ represent a pair of brackets.
It is then shown that $(u+ X + v)^* = (Nv)^*N (uN)^*$, where $N \in K$ is the least solution of $y \geq (X+uyv)^*$ corresponding to $D_1(X)$.
Except for the balanced bracket occurrences in $N$, in $(Nv)^*N (uN)^*$ all occurrences of the closing bracket $v$ are to the left of all occurrences of the opening bracket $u$.
This is similar to the normal form $\nf w \in Q_1^*P_1^* \cup \{0\}$ in the polycyclic monoid $P_1'$
of Section \ref{sec-polycyclic} with $P_1=\{u\}$ and $Q_1=\{v\}$, i.e.~the normal
forms on $\{u, v\}^*$ modulo the congruence generated by $uv = 1$, and its extension
to $\nf{w}\in Q_1^*X^*V_1^*\cup\{0\}$ for $w\in P_1'[X]$ where elements of $X$
commute with those of $P_1\cup Q_1$.

\begin{proposition}\label{prop-least-solutions}
	Let $K$ be a Kleene algebra and $u, x, v \in K$.
	If $y \geq (x+uyv)^*$ has a least solution $N$, then $N = (x+uNv)^*$ and $N$
        is the least solution of $y \geq 1+x+uyv+yy$.
	If $y \geq 1+x+uyv+yy$ has a least solution $D$, then $D=1+x+uDv+DD$ and $D$ is the least solution of $y \geq (x+uyv)^*$.
\end{proposition}

\begin{proof}
  Let $f$ and $h$ be defined by $f(y) = x + uyv$ and $h(y) = 1 + x + uyv + yy$.  (i)
  If $y\geq h(y)$, then $y\geq f(y)$ and $y\geq 1+yy$, hence $y\geq y^*$ by axioms of
  Kleene algebra, and so $y\geq y^* \geq f(y)^*$ by monotonicity of ${}^*$.  (ii)
  Conversely, if $y\geq f(y)^*$, then $f(y)^*\geq h(f(y)^*)$, because
  \[ h(f(y)^*) \leq 1 + x + uyv + f(y)^*f(y)^* \leq f(y) + f(y)^* \leq f(y)^*. \]
  It follows that if $N$ is the least solution of $y \geq f(y)^*$, then by (i), any
  solution of $y \geq h(y)$ satisfies $y \geq N$, and by (ii), $f(N)^*$ is a
  solution of $y \geq h(y)$, so $f(N)^*= N$ is the least solution of $y \geq h(y)$.

  If $D$ is the least solution of $y \geq h(y)$, then by (ii), any solution of $y\geq
  f(y)^*$ satisfies $y \geq f(y)^* \geq D$, and by (i), $D \geq f(D)^*$. Hence $D$ is
  the least solution of $y\geq f(y)^*$. Then $D=f(D)^*$ and hence $D = DD = 1 + f(D) + DD = h(D)$.
\end{proof}

\begin{theorem} \label{thm-Nxuv}
	Let $K$ be a Kleene algebra and $x, u, v \in K$.
	If $y \geq (x+uyv)^*$ has a least solution $N$ in $K$, then $(u+x+v)^* = (Nv)^*N (uN)^*$.
\end{theorem}

\begin{proof}
	Let $N = \mu y.(x+uyv)^*$ and $n=(u+x+v)^*$.
	We first show $N \leq n$, by showing that $n$ solves $(x+uyv)^* \leq y$.
	By monotonicity of $+, \cdot$, and \Star,
	\[
		x+unv \leq x + un^*v \leq n+nn^*n = (1+ nn^*)n = n^*n \leq n^*=n,
	\]
	hence $(x + unv)^* \leq n^* = n$.
	So $N \leq n$, from which
	\begin{eqnarray*}
		(Nv)^*N (uN)^*  &\leq& (u+x+v)^*
	\end{eqnarray*}
	follows using $u, v, N \leq n$ and $(nn)^* = n^* = n = nnn$. 

	Now consider the reverse inequality, $(u+x+v)^* \leq (Nv)^*N (uN)^*$:
	As  $(x+uNv)^* = N$ by Proposition \ref{prop-least-solutions}, we have $(x+uNv)N+1 \leq N$.
	Using this and Kleene algebra identities like $(ab)^*a = a(ba)^*$, we show that $(Nv)^*N (uN)^*$ solves $(u+x+v)z + 1 \leq z$ in $z$:
	\begin{eqnarray*}
		\lefteqn{(u+x+v)(Nv)^*N (uN)^* + 1 }
	\\	&=&	(u+x+v)N (vN)^*(uN)^* + 1
	\\	&=&	uN (vN)^*(uN)^*+xN (vN)^*(uN)^*+vN (vN)^*(uN)^* + 1
	\\	&=&	uN (1+vN (vN)^*)(uN)^*+xN (vN)^*(uN)^*+vN (vN)^*(uN)^* + 1
	\\	&=&	uN (uN)^*+uNvN (vN)^*(uN)^*+xN (vN)^*(uN)^*+vN (vN)^*(uN)^* + 1
	\\	&=&	(x+uNv)N (vN)^*(uN)^*+ uN (uN)^*+vN (vN)^*(uN)^* + 1
	\\	&=&	(x+uNv)N (vN)^*(uN)^*+ (1+vN (vN)^*)(uN)^*
	\\	&=&	(x+uNv)N (vN)^*(uN)^*+ (vN)^*(uN)^*
	\\	&=&	((x+uNv)N+1)(vN)^*(uN)^*
	\\	&\leq&	N (vN)^*(uN)^*
	\\	&=&	(Nv)^*N (uN)^*.
	\end{eqnarray*}
	Since $(u+x+v)^*$ is the least solution of $(u+x+v)z + 1 \leq z$, the claim $(u+x+v)^* \leq (Nv)^*N (uN)^*$ is shown.
\end{proof}

It is worth noticing that these results are generic to Kleene algebras and do not require the \Star-continuity property.
They are all conditioned on the \emph{existence} of the relevant least-fixed-points,
and it is for existence that \Star-continuity will come into play.

\subsection{Normal form theorems}
\label{sec-normal-forms}

Let $\CA = \Z{S, A, F}$ be an automaton with $A = U+X+V$ as in Theorem~\ref{thm-aut},
representing an element $\phi=L(\CA) = SA^*F$ of $K\xR C_2'$.  We first show that
there is a least solution of $y \geq (UyV+X)^*$ in $\Mat nn{K\xR C_2'}$, which is
related to Dyck's context-free language $D \subseteq \{U, X, V\}^*$ of balanced
strings of matrices, with $U$ as ``opening bracket'' and $V$ as ``closing bracket''.
Namely, if concatenation is interpreted by matrix multiplication and the empty
sequence as unit matrix, $D$ becomes a context-free subset of $\mat nn{(K\xR C_2')}$
and the least solution of $y\geq (UyV+X)^*$ its least upper bound.

\begin{lemma} \label{lem-N-matrix-C2'}
	Let $K$ be an $\CR$-dioid, $n \in \N$,  $X \in \mat nn{(\zc{K}{C_2'})}$, $U \in \{0, b, p\}^{n\times n}$ and $V \in \{0, d, q\}^{n\times n}$.
	In $\Mat nn{K\xR C_2'}$,
	\begin{equation} \label{eqn-N-C2'}
		y \geq (UyV+X)^*
	\end{equation}
	has a least solution, namely $N := b(Up+X+qV)^*d$, and $N \in \mat nn{(\zc K{C_2'})}$.
\end{lemma}
When multiplying $b, d, p, q$ with $n\times n$-matrices, we identify them with
corresponding diagonal matrices.\footnote{~The proof will show that $N$ is the least
  upper bound of a context-free set $D$ of $n\times n$-matrices over $Z=\zc K{C_2'}$
  (with $DD\subseteq D$) and the least solution of the matrix inequation $y\geq
  1+X+UyV+yy$. Alternatively, by Theorem \ref{thm-centralizer}, \ref{item-4}, $Z$ is a $\CC$-dioid,
  and by \cite{Leiss2016}, its $n\times n$ matrix semiring also is. Hence $D$ has a
  least upper bound $\sum D$ and $(\sum D)(\sum D)= \sum (DD)\leq \sum D$. Since
  $U,V$ are not matrices over $Z$, one needs additional arguments to show that $\sum
  D$ is the least solution of the matrix inequation in $\Mat nnZ$ and least in $\Mat
  nn{K\xR C_2'}$. Our proof here is more elementary and uses properties of $\CR$-dioids
  only.}

\begin{proof}\setcounter{claim}{0}
  Let $D$ and $D'$ be the Dyck languages over $\{U, X, V\}$ and $\{Up, X, qV\}$ with
  brackets $U, V$ and $Up, qV$, respectively.  By interpreting concatenation as matrix
  multiplication and the empty sequence as unit matrix, elements of $D$ and $D'$
  belong to $\Mat nn{K\xR C_2'}$.  To simplify the notation, we write $T$ for
  $Up+X+qV$ and $Z$ for $\zc{K}{C_2'}$.

  \begin{Claim} \label{claim-Lemma-1}
    Every $A \in D$ evaluates in $\Mat nn{K\xR C_2'}$ to an element of $\mat nnZ$.
  \end{Claim}

	\begin{subproof}
		This is clear for $A=1$ and $A=X$, and if $A,B\in D$ evaluate to $A, B \in \mat nnZ$, then $AB \in \mat nnZ$, because $Z$ is a semiring.
		Finally, consider $A=UBV$ with $B \in \mat nnZ$.
		Since elements of $Z$ and $C_2'$ commute with each other in $K\xR C_2'$, we have
		\[
		(UBV)_{ij} = \sum_{k,l=1}^n U_{ik}(B_{kl}V_{lj}) = \sum_{k,l=1}^n B_{kl}(U_{ik}V_{lj}),
		\]
		and since $U_{ik} \in \{0, b, p\}$ and $V_{lj} \in \{0, d, q\}$, we obtain $U_{ik}V_{lj} \in \{0, 1\}$, hence $(UBV)_{ij} \in Z$, and so $A \in \mat nnZ$.
        \end{subproof}
		It follows that $bAd=A=pAq$ for each $A\in D$ and $\sum(\{U, X, V\}^m
                \cap D) \in \mat nnZ$ for each $m\in
                \N$.
	\begin{Claim} \label{claim-Lemma-1b}
		$bT^md = \sum(\{U, X, V\}^m \cap D)$ and
                $bT^md\leq T^m$, for each $m \in \N$.
	\end{Claim}
        \begin{subproof}
		Let $A' \in D'$ be obtained from $A \in D$ by replacing factors $U$
                by $Up$ and factors $V$ by $qV$. Then as matrices, $A'=A$: clearly
                $1'=1$ and $X'=X$, and by induction, for $A,B\in D$, $(AB)' = A'B' =
                AB$ and $(UAV)' = UpA'qV = UpAqV = UAV$, as $A$ belongs to $\mat nnZ$
                by claim~\ref{claim-Lemma-1}.  Moreover, if $A \in D \cap \{U, X,
                V\}^m$, then the matrix value of $A' \in \{Up, X, qV\}^m \cap D'$ is
                a summand of $T^m$ and thus $A=A' \leq T^m$. By monotonicity,
                  $A=bAd\leq bT^md$. It follows that
                \[ \sum(\{U, X, V\}^m \cap D)\leq T^m \quad\textrm{and}\quad
                   \sum(\{U, X, V\}^m \cap D)\leq bT^md.
                   \]
		To show the reverse of the second inequation, let $A' \in \{Up, X,
                qV\}^m$ be a summand of $T^m=(Up + X + qV)^m$ that is not obtained
                from any $A \in \{U, X, V\}^m \cap D$ by this substitution.  Then
                $bA'd = 0$, because $A' \in (D'qV)^*D'(UpD')^* \setminus D'$ and $b,
                d$ commute with factors from $D'$ (with values in $\mat nnZ$), so in
                $bA'd$, $b$ can be moved over factors to the right, until it meets
                $q$ and gives $bq=0$, or $d$ can be moved over factors to the left
                until it meets $p$ and gives $pd=0$.  It follows that $bT^md \leq
                \sum(\{U, X, V\}^m \cap D) \leq T^m$.
	\end{subproof}

	By ${}^*$-continuity, claim~\ref{claim-Lemma-1b} implies that the set $D$ of
        matrices obtained from the context-free language $D\subseteq\{U,X,V\}^*$ has
        a least upper bound in $\Mat nn{K\xR C_2'}$:
	\begin{eqnarray*}
		N = bT^*d 
			&=& \Sum{bT^md}{m\in\N} \\
			&=& \Sum{D \cap \{U, X, V\}^m}{m\in\N} = \sum D.
	\end{eqnarray*}

	\begin{Claim}
		$N \in (\zc K{C_2'})^{n\times n}$.
	\end{Claim}
	\begin{subproof}
		We have seen $bT^md\in \mat nnZ$ for each $m\in\N$. So for each $c\in
                C_2'$, $c(bT^md) = (bT^md)c$ and
                \[ cN = cbT^*d = \Sum{cbT^md}{m\in\N}) = \Sum{bT^mdc}{m\in\N}) = bT^*dc = Nc,
                \]
                since $\Mat nn{K\xR C_2'}$ is \Star-continuous. It follows that each
                entry of $N$ commutes with $c$. 
	\end{subproof}

	\begin{Claim}\label{claim-Lemma-1d}
		$N$ is the least solution of $y \geq (UyV+X)^*$ in $\Mat nn{K\xR C_2'}$.
	\end{Claim}
	\begin{subproof}
		We show that $N$ is the least solution of $y \geq 1+X + UyV + yy$ and
                apply Proposition \ref{prop-least-solutions}.  By claim~\ref{claim-Lemma-1b},
                we get $1+X \leq N$ and since $bT^md$ is a finite sum of balanced
                sequences of length $m$ over $\{U,X,V\}$, by distributivity $UbT^mdV$ is a sum of
                balanced sequences of length $m+2$, hence
		\[
			UbT^mdV \leq \sum(\{U, X, V\}^{m+2} \cap D) = bT^{m+2}d \leq N.
		\]
		Thus by ${}^*$-continuity, $UNV = UbT^*dV =
                \sum\setof{UbT^mdV}{m\in\N} \leq N$.
		It remains to show $NN \leq N$.
		By \Star-continuity,
		\begin{eqnarray*}
			NN	&=&	\sum_{k\in\N}bT^kd N = \sum_{k,l\in\N}{bT^kdbT^ld}.
		\end{eqnarray*}
		By claim~\ref{claim-Lemma-1} and claim~\ref{claim-Lemma-1b}, $bT^kd
                \in \mat nnZ$, so $(bT^kd)bT^ld = b(bT^kd)T^ld$, and $bT^kd \leq
                T^k$, whence
                \[ NN = \sum_{k,l\in\N}b(bT^kd)T^ld \leq \sum_{k,l\in\N}bT^kT^ld = N. \]

		Therefore, $N$ is a solution of $y \geq 1 + X + UyV + yy$.  To show
                that it is the least solution, suppose $y \in \Mat nn{K\xR C_2'}$
                satisfies $y \geq 1 + X + UyV + yy$.  As $N=\sum D$, it is sufficient
                to show $A \leq y$ for each $A \in D$.  This is clear for $1$ and
                $X$, and if $A, B \in D$ satisfy $A, B \leq y$, then $UAV \leq UyV
                \leq y$ and $AB \leq yy \leq y$ by monotonicity.  So $y$ is an upper
                bound of $D$.
	\end{subproof}

	By the last two claims, the Lemma is proven.
\end{proof}


\begin{example} \label{expl-2}
	In the most simple case $n=1$, with $\Mat nn{K\xR C_2'} \simeq K\xR C_2'$, suppose $U=b, V=d$ and $X = x \in K$.
	Then $N = b(bp+x+qd)^*d = \sum D$ for Dyck's language $D \subseteq \{b, x, d\}^*$.
	The proof shows $N=\sum D \in \zc{K}{C_2'}$.
\qee
\end{example}

\begin{theorem}[First Normal Form]\label{thm-nf-C'}
	Let $K$ be an $\CR$-dioid.
	For each $\phi \in K\xR C_2'$ there are $n \in \N$, $S \in \B^{1\times n}$, $F \in \B^{n\times 1}$,
	$U \in \{0, b, p\}^{n\times n}$, $V \in \{0, d, q\}^{n\times n}$ and $X \in K^{n\times n}$ such that
	\[
		\phi = S(NV)^*N (UN)^*F,
	\]
	where $N \in (\zc{K}{C_2'})^{n\times n}$ is the least solution of $y \geq (UyV+X)^*$ in $\Mat nn{K\xR{C_2'}}$.
\end{theorem}
For $n=1$, $N$ commutes with $U$ and $V$, so $(NV)^kN (UN)^l = V^kNU^l$, and by \Star-continuity, $(NV)^*N (UN)^* = V^*NU^*$.
This is related to the normal form for the extension $P_m'[X]$ of the polycyclic monoid $P_m'$ in Section \ref{sec-polycyclic}.

\begin{proof}
	By definition of $K\xR{C_2'}$, there is $R \in \CR(K\times C_2')$ such that $\phi = [R]$.
	As in Theorem \ref{thm-aut}, by induction on $R$ one constructs an automaton $\Z{S, A, F}$ with
	\[
		\phi = [R] = L(\Z{S, A, F}) = SA^*F
	\]
	and a transition matrix $A \in (K\xR C_2')^{n\times n}$ of the form $A = U+X+V$ where $U \in \{0, b, d\}^{n\times n}$, $X \in K^{n\times n}$
	and $V \in \{0, d, q\}^{n\times n}$, for some $n$.
	By Lemma \ref{lem-N-matrix-C2'}, $y \geq (UyV + X)^*$ has a least solution $N$ in $\Mat nn{K\xR C_2'}$, and
	\[
		N \in (\zc K{C_2'})^{n\times n}.
	\]
	By Theorem \ref{thm-Nxuv}, this $N$ allows us to write $A^*$ as
	\[
		A^*=(U+X+V)^* = (NV)^*N (UN)^*
	\]
	and obtain the normal form $\phi = [R] = S A^* F = S (NV)^* N (UN)^* F$.
\end{proof}

While Theorem \ref{thm-nf-C'} gives a generic normal form for an element $\phi$ of $K\xR
C_2'$, it is not straightforward to compute the matrix $N$ occurring in the normal
form of $\phi$. The following example demonstrates how
 $N$ is obtained from an automaton for $\phi$ through the construction of Lemma \ref{lem-N-matrix-C2'}.
In Section \ref{sec-combining-nfs} we will show how to compute a normal form
inductively from a regular expression $\phi$. 

\begin{example}\label{expl-3}
	Let $P_2=\{\p0, \p1\}$, $Q_2=\{\q0, \q1\}$, and $K = \CR\{a, b\}^*\xR C_2'$.
	The element $\phi=(a\p1)^*(\q1b)^* \in K$ is represented as $\phi = L(\CA) = SA^*F$ by the automaton $\CA=\Z{S, A, F}$ of Figure~\ref{fig-1} with initial state 1 and accepting state 3.
	\begin{figure}[ht]
	\qquad\quad\begin{minipage}{7.1cm}
	\(
		\left\langle
		\begin{pmatrix}
			1	& 0	& 0	& 0
		\end{pmatrix}, \begin{pmatrix}
			0	& a	& 1	& 0 \\
			\p1	& 0	& 0	& 0 \\
			0	& 0	& 0	& \q1 \\
			0	& 0	& b	& 0
		\end{pmatrix}, \begin{pmatrix}
			0 \\
			0 \\
			1 \\
			0
		\end{pmatrix}\right\rangle
	\)
	\vskip-1mm
	\caption{$\CA = \Z{S, A, F}$ \label{fig-1}}
	\end{minipage}
	\hskip1.5cm
	\begin{minipage}{3.6cm}
	\vskip-2mm\quad\begin{tikzcd}[row sep=1.2cm, column sep=1.2cm]
		  1 \arrow[d, "a" left, shift right=0.6ex]\arrow[r,"1" above]
                  & 3 \arrow[d, "\q1" left, shift right=.6ex]
               \\ 2\arrow[u, "\p1" right, shift right=0.6ex]
                  & 4\arrow[u, "b" right, shift right=0.6ex]
		\end{tikzcd}
	\vskip-1mm
	\caption{Graph of $A$\label{fig-2}}
	\end{minipage}
	\end{figure}

	The iteration $A^*$ of $A$ calculated
	using the formula (\ref{mat-star}) can be read off from the graph: the entry $(A^*)_{i,j}$ describes the labellings on paths from node $i$ to node $j$.
	Hence, with $\bar{a} = a\p1$ and $\bar{b} = \q1 b$, we have
	\[
		A^* = \begin{pmatrix}
			\bar{a}^*	& \bar{a}^*a		& \bar{a}^*\bar{b}^*		& \bar{a}^*\bar{b}^*\q1 \\
			\p1 \bar{a}^*\	& 1 + \p1\bar{a}^*a\	& \p1\bar{a}^* \bar{b}^*\	& \p1\bar{a}^* \bar{b}^*\q1\ \\
			0		& 0			& \bar{b}^*			& \bar{b}^*\q1 \\
			0		& 0			& b\bar{b}^*			& 1 + b\bar{b}^*\q1
		\end{pmatrix}.
	\]

	To obtain the normal form $(NV)^* N (UN)^*$ of $A^*$, split $A$ as $U+X+V$ with
	\[
		U = \begin{pmatrix}
			0	& 0	& 0	& 0 \\
			\p1	& 0	& 0	& 0 \\
			0	& 0	& 0	& 0 \\
			0	& 0	& 0	& 0
		\end{pmatrix}, \quad X = \begin{pmatrix}
			0	& a	& 1	& 0 \\
			0	& 0	& 0	& 0 \\
			0	& 0	& 0	& 0 \\
			0	& 0	& b	& 0
		\end{pmatrix}, \quad V = \begin{pmatrix}
			0	& 0	& 0	& 0 \\
			0	& 0	& 0	& 0 \\
			0	& 0	& 0	& \q1 \\
			0	& 0	& 0	& 0
		\end{pmatrix}.
	\]
	To determine $N = \p0(U\p1+X+\q1V)^*\q0$, let $\widetilde A = (U\p1+X+\q1V)$ and read off ${\widetilde A}^*$ from the graph of $\widetilde A$,
	obtaining a copy of $A^*$ with $\tilde{a} = a\p1^2, \tilde{b} = \q1^2b, \p1^2, \q1^2$ instead of $\bar{a}, \bar{b}, \p1, \q1$, respectively.
	The entries of $N$ are then $$N_{i,j} = \p0 (\widetilde A^*)_{i,j}\q0.$$
	The resulting matrix is as follows, writing $\widehat L$ for $\sum L$ with $L = \setof{a^nb^n}{n\in\N}$,
	\[
	\renewcommand{\arraycolsep}{4pt}
		N = \p0 \begin{pmatrix}
			\tilde{a}^*		& \tilde{a}^*a		& \tilde{a}^*\tilde{b}^*	& \tilde{a}^*\tilde{b}^*\q1^2 \\
			\p1^2\tilde{a}^*	& 1 + \p1^2\tilde{a}^*a	& \p1^2\tilde{a}^*\tilde{b}^*	& \p1^2\tilde{a}^* \tilde{b}^*\q1^2 \\
			0			& 0			& \tilde{b}^*			& \tilde{b}^*\q1^2 \\
			0			& 0			& b\tilde{b}^*			& 1 + b\tilde{b}^*\q1^2
		\end{pmatrix} \q0 = \begin{pmatrix}
			1	& a	& \widehat L	& a\widehat L \\
			0	& 1	& \widehat Lb	& \widehat L \\
			0	& 0	& 1		& 0 \\
			0	& 0	& b		& 1
		\end{pmatrix}.
	\]
	For example, $N_{1,3} = \p0\tilde{a}^*\tilde{b}^*\q0 = \p0(a\p1^2)^*(\q1^2b)^*\q0 = \widehat L$ is calculated as in Example \ref{expl-1}.
	It follows that
	\[
		NV = \begin{pmatrix}
			0	& 0	& 0	& \widehat L \q1 \\
			0	& 0	& 0	& \widehat Lb \q1 \\
			0	& 0	& 0	& \q1 \\
			0	& 0	& 0	& b \q1
		\end{pmatrix}, \qquad UN = \begin{pmatrix}
			0	& 0	& 0		& 0 \\
			\p1	& \p1 a	& \p1\widehat L	& \p1 a\widehat L \\
			0	& 0	& 0		& 0 \\
			0	& 0	& 0		& 0
		\end{pmatrix},
	\]
	which imply $(NV)^* = 1 + NV(b\q1)^*$ and $(UN)^* = 1 + (\p1 a)^*UN$.
	By matrix multiplication, one obtains the normal form $(NV)^* N (UN)^* = A^*$.

	To determine $N$, one can also use that $N$ is the least solution of $y \geq (UyV+X)^*$ in $\Mat 44K$, hence $N = (UNV+X)^*$.
	Let $e_i$ be the unit column vector with 1 in the $i$-th row, 0 else, $e_i'$ its transpose row vector.
	Then $e_ie_j'$ is the $4\times 4$-matrix with 1 at $(i, j)$, 0 else, and $e_i'e_j$ the $1\times 1$-matrix with entry $\delta_{i,j}$.
	Since
	\[
		UNV = (e_2\p1e_1')(\sum_{1\leq i,j\leq 4} e_iN_{i,j} e_j')(e_3\q1e_4') = e_2\p1 N_{1,3}\q1e_4' = e_2N_{1,3}e_4',
	\]
	the graph of $X+UNV$ is that of $X$ with additional edge $2\edge{N_{1,3}}4$, from which one can read off $(X+UNV)^*$ as
	\[
	\renewcommand{\arraycolsep}{2ex}
		(X+UNV)^* = \begin{pmatrix}
			1	& a	& 1+aN_{1,3}b	& aN_{1,3} \\
			0	& 1	& N_{1,3}b	& N_{1,3} \\
			0	& 0	& 1		& 0 \\
			0	& 0	& b		& 1
		\end{pmatrix} = N.
	\]
	Since $N$ is the least solution of $y \geq (UyV+X)^*$, $N_{1,3}$ is the least solution of $y_{1,3} \geq 1 + a y_{1,3}b$,
	i.e.~$\mu x(1+axb) = \sum L$ for $L = \setof{a^nb^n}{n\in\N} \in \CC K$, leading to the matrix $N$ shown above.
\qee
\end{example}

\subsection{Reduced normal form}\label{sec-reduced-nf}

We conjectured in \cite{HopkinsLeiss23} that the normal form
$S(NV)^*N(UN)^*F$ for $\phi \in K\xR C_2'$ given in Theorem~\ref{thm-nf-C'} can be
simplified to $SNF$ for elements $\phi \in \zc{K}{C_2'}$. We can now prove this under
the additional assumption that $K$ is non-trivial and has no zero divisors. 

\begin{lemma} \label{lem-embedding-KxC}
	Let $m \geq 2$, $g:C_m' \to C_2'$ the $\CR$-embedding of Lemma \ref{lem-recoding-C'}, and $K$ an $\CR$-dioid.
	There is an $\CR$-embedding $\ol{\, \cdot\,}:K \xR C_m' \to K \xR C_2'$, given by
	\[
		\ol{[R]} = \Sum{a \cdot g(b)}{(a,b) \in R} \quad \textrm{for }R \in \CR(K\times C_m),
	\]
	which maps $\zc{K}{C_m'}$ to $\zc{K}{C_2'}$; for $m=2$, it is the identity on
        $\zc{K}{C_2'}$.
\end{lemma}

\begin{proof}
Let $\ol{\,\cdot\,}$ be the induced injective $\CR$-morphism $h_\fg$ for the embeddings
$f=\Id_K$ and $g:C_m'\to C_2'$ in
\[
\begin{tikzcd}[row sep=1.6cm, column sep=1.1cm]
  K_1 \arrow[r,"\top_1'"] \arrow[d,"f"]
  & K_1\xR C_m' \arrow[d,dotted,"{h_{\fg}}"]
  & \arrow[l,"\top_2'" above]\arrow[d,"g"] C_m'
\\    K \arrow[r,"\top_1"] & K\xR C_2' & \arrow[l,"\top_2" above] C_2'.
\end{tikzcd}
\]
according to Corollary \ref{cor-injectivity}.
For $R\in \CR(K\times C_m')$, the element $[R]'\in K\xR C_m'$ is mapped to
\[ \ol{[R]'} = [\fg(R)] = \sum\setof{f(a)\otimes g(b)}{(a,b)\in R} . \]
	By Theorem \ref{thm-centralizer} \ref{item-3}, each element of $\zc{K}{C_m'}$ is the congruence class $[R]'$ of some relation $R \in \CR(K\times C_m')$ with $R \subseteq K \times \{0, 1\}$.
	Since $g(0)=0$ and $g(1)=1$, we have $[\fg(R)] = [R]$.
	Hence $\ol{\,\cdot\,}$ restricts to an $\CR$-morphism $\ol{\,\cdot\,}:\zc{K}{C_m'} \to \zc{K}{C_2'}$.
	For $m=2$, this is the identity on $\zc{K}{C_2'}$, since $\top_1=\top_1'$, $\top_2=\top_2'$ and
        $\Id_K\!\times g$ leaves $R$ fixed.
\end{proof}

\begin{corollary}[Reduced Normal Form] \label{cor-nf-C'}
	Let $K$ be a non-trivial $\CR$-dioid without zero divisors. Let $\phi = SA^*F \in K\xR C_2'$ with $A=U+X+V$ and $n, S, F, U, X, V$ and $N$ as in Theorem \ref{thm-nf-C'}.
	If $\phi \in \zc K{C_2'}$, then $\phi=SNF$.
\end{corollary}

\begin{proof}
	Suppose $\phi = SA^*F \in \zc{K}{C_2'}$.
	Since $\phi$ is a finite sum of entries of ${A^*}$, by Corollary \ref{cor-zc-bounded}, all summands belong to $\zc{K}{C_2'}$.
	Therefore, $\phi=SA^*F = SNF$ is shown if for all $i, j < n$
	\begin{eqnarray} \label{eqn-SA*F}
		{A^*}_{i,j} \in \zc{K}{C_2'}	&\Longrightarrow&	{A^*}_{i,j} = N_{i,j}.
	\end{eqnarray}

	Let $\ol{\,\cdot\,}:K \xR C_2' \to K \xR C_2'$ be the
        $\CR$-morphism of Lemma \ref{lem-embedding-KxC} 
	On $\zc{K}{C_2'}$, it is the identity.
	Applying $\ol{\,\cdot\,}$ entrywise to matrices we get
	\begin{eqnarray*}
		\ol{A^*}	&=& \ol{(NV)^* N (UN)^*}
	\\			&=& (\ol{N}\ol{V})^*\ol{N}(\ol{U}\ol{N})^*
	\\			&=& (N\ol{V})^* N (\ol{U}N)^*.
	\end{eqnarray*}
	Notice $\ol{U} \in \{\ol{0}, \ol{b}, \ol{p}\}^{n\times n} = \{0, bp, bp^2\}^{n\times n}$ and $\ol{V} \in \{\ol{0}, \ol{d}, \ol{q}\}^{n\times n} = \{0, qd, q^2d\}^{n\times n}$.
	By Lemma~\ref{lem-recoding-C'} and $b,d$ as diagonal matrices, $b\ol{V} = 0 = \ol{U}d$, so $b(N\ol{V}) = Nb\ol{V} = 0 = \ol{U}dN= (\ol{U}N)d$, hence $b(N\ol{V})^* = b$ and $(\ol{U}N)^*d = d$.
	For $(A^*)_{i,j} \in \zc{K}{C_2'}$ this gives
	\begin{eqnarray*}
		(A^*)_{i,j}	&=& (\ol{A^*})_{i,j} = b(\ol{A^*})_{i,j}d
	\\			&=& (b{\ol{A^*}}d)_{i,j} = (b(N\ol{V})^* N (\ol{U}N)^*d)_{i,j}
	\\			&=& (b N d)_{i,j}
	\\			&=& N_{i,j}.
	\end{eqnarray*}
	We thus have shown (\ref{eqn-SA*F}).
\end{proof}

Notice that in the useful cases where $K = \CR M$ for a monoid $M$, indeed $K$
is non-trivial and has no zero divisors.

In the special case of $\zc{\CR X^*}{C_m'}$, the elements of the centralizer of $C_m'$ have previously been chararcterized as follows:
\begin{theorem}[Corollary 28 of \cite{Leiss22}]
	For $m > 2$ and $\phi \in \CR{X^*}\xR C_m'$, we have $\phi \in \zcr{X^*}{C_m'}$
	iff there is a regular expression $r$ over $X \dotcup (\Delta_m \setminus \{\p0, \q0\})$ such that $\phi = \p0r\q0$.
\end{theorem}

To prove this, one codes the $m > 2$ bracket pairs by the two pairs $p_1, q_1$ and $p_2, q_2$ to get a regular expression $r$ in $\ol{p_i}=p_1p_2^{i+1}$
and $\ol{q_j}=q_2^{j+1}q_1$, and then has $p_0, q_0$ as a fresh bracket pair to eliminate the unbalanced strings using $p_0 r q_0$.
One can do the same for $m=2$:

\begin{quote}
	We have $\phi \in \zc{\CR X^*}{C_2'}$ iff there is a regular expression $r$ over $X \dotcup \Delta_2$
	with $\p0$ only as part of $\p0\p1$ and $\q0$ only as part of $\q1\q0$, such that $\phi = \p0 r\q0$.
\end{quote}

For any $m \geq 2$, the first normal form theorem \ref{thm-nf-C'} holds as well with $C_m'$ instead of $C_2'$.
If the automaton $\Z{S, A, F}$ for $\phi$ has no transitions under $\p0$ and $\q0$,
then $\phi \mapsto \p0\phi\q0$ is a projection on the centralizer:

\begin{corollary} \label{cor-1}
	Suppose $\phi = SA^*F \in K\xR C_m'$ is represented by an automaton $\Z{S, A, F}$ not using $\p0, \q0$,
	i.e.~$U \in \{0, \p1, \ldots, \p{m-1}\}^{n \times n}$ and 
        $V \in \{0, \q1, \ldots, \q{m-1}\}^{n \times n}$ in $A = U + X + V$.
	If $S(NV)^* N (UN)^*F$ is the normal form of $\phi$, then
	\[
		\p0 \phi \q0 =SNF \in \zc{K}{C_m'}.
	\]
\end{corollary}

\begin{proof}
	By the assumption on $U$ and $V$, for the diagonal matrix versions of
        $\p0,\q0$ we have $\p0V=0=U\q0$, and since $N$ commutes with $\p0$ and $\q0$, we get $\p0 (NV)^* = \p0$ and $(UN)^*\q0 = \q0$.
	Hence
	\[
		\p0 A^*\q0 = \p0 (NV)^* N (UN)^*\q0 = \p0 N\q0 = N,
	\]
	and thus $\p0\phi\q0 = \p0 SA^*F\q0 = S\p0 A^*\q0 F = SNF \in \zc{K}{C_m'}$.
\end{proof}

\subsection{Second normal form}

Corollary~\ref{cor-1} can be extended by admitting that $\phi = SA^*F \in K\xR{C_m'}$
is given by an automaton $\Z{S, A, F}$ whose transition matrix $A$ contains transitions by $\q0\p0$ in addition to those by elements of $K$ and $\Delta_m \setminus \{\p0, \q0\}$.
This is useful to combine representations $\p0 r_i\q0 = \sum L_i$ of $L_i \in \CC X^*$, $i = 1, 2$,
in $\CR X^*\xR C_2'$ to a representation $\p0 r_1\q0\p0 r_2\q0 = (\sum L_1)(\sum L_2) = \sum(L_1L_2)$ of $L_1L_2$, as will be exemplified below.

\begin{theorem}[Second Normal Form]\label{thm-nf2}
	Let $K$ be an $\CR$-dioid, $m \geq 2$ and $\phi \in K\xR C_m'$ be given in matrix form $\phi = S(U + X + V + W\pi)^*F$, where $\pi = \q0\p0$ and for some $n \geq 0$,
	\[
		\begin{array}{rcl}
			S	&\in&	\{0, 1\}^{1\times n},
		\\	F	&\in&	\{0, 1\}^{n\times 1},
		\end{array} \quad \begin{array}{rcl}
			X	&\in&	K^{n\times n},
		\\	W	&\in&	\{0, 1\}^{n\times n},
		\end{array} \quad \begin{array}{rcl}
			U	&\in&	\{0, \p1, \ldots, \p{m-1}\}^{n\times n},
		\\	V	&\in&	\{0, \q1, \ldots, \q{m-1}\}^{n\times n}.
		\end{array}
	\]
	Then there is a least solution $N$ of $y\geq (UyV+X)^*$ in $\Mat
        nn{K\xR C_m'}$, and
        \[
		\p0\phi\q0 = S N (WN)^*F\ \in \zc{K}{C_m'}.
	\]
\end{theorem}

\begin{proof} 
	Let $A=U+X+V$.
	By Theorem \ref{thm-nf-C'}, there is $N = \mu y.(UyV+X)^* \in (\zc{K}{C_m'})^{n\times n}$ with
	\[
		A^* = (U+X+V)^* = (NV)^* N (UN)^*.
	\]
	As in the proof of Corollary~\ref{cor-1}, we obtain
	\[
		\p0 A^* \q0 = \p0 (NV)^* N (UN)^* \q0 = \p0 N \q0 = N,
	\]
	and therefore in the Kleene algebra $\Mat nn{K\xR C_m'}$, using identities $(a+b)^*=a^*(ba^*)^*$ and $(ab)^*a = a(ba)^*$ of Kleene algebra,
	\begin{eqnarray*}
		\p0 (A+W\pi)^* \q0	&=&	\p0 A^*(W\pi A^*)^* \q0
	\\				&=&	\p0 A^*(\q0 W\p0 A^*)^*\q0
	\\				&=&	\p0 A^* \q0 (W\p0 A^*\q0)^*
	\\				&=&	N (WN)^* \in (\zc{K}{C_m'})^{n\times n}.
	\end{eqnarray*}
	Because $S, N, W$ and $F$ commute with $\p0$ and $\q0$, it follows that
	\[
		\p0\phi\q0 = S\p0 (A+W\pi)^* \q0 F = S N (WN)^* F \in \zc{K}{C_m'}.
	\]
	Notice also that $\pi \phi\pi = \q0\p0\phi\q0\p0 = \pi \p0\phi\q0$.
\end{proof}

\begin{example} \label{expl-3ff}
	Consider $\phi=(a\p1)^*(\q1b)^* \in K := \CR\{a,b\}^*\xR C_2'$ of Example \ref{expl-3} and its automaton $\Z{S, A, F}$ with $A=U+X+V$ and graph as shown in Figure \ref{fig-2}.
        We have seen that $\p0\phi\q0 = \sum L$ represents $L= \setof{a^nb^n}{n\in\N}
        \in \CC\{a,b\}^*$ in $K$. By Corollary \ref{cor-1}, $\p0\phi\q0$ is the
        projection of $\phi$ to the centralizer $Z_{C_2'} K$.
	Using $\pi=\q0\p0$, we claim that the projection of $\psi = \phi\pi\phi$ to the centralizer represents $LL$ in $K$, i.e.~$\p0\psi\q0 = \sum(LL)\in Z_{C_2'} K$. 
	To obtain an automaton $\Z{\tilde S,\tilde A,\tilde F}$ for $\psi$,
        connect the graph of $A$ with a copy of itself by an edge labelled by
        $\pi$, to get the graph of $\tilde A$ shown in Figure~\ref{fig-3}.

	\begin{figure}[ht]
	\begin{minipage}{14cm}
	\[
	\begin{tikzcd}[row sep=1.2cm, column sep=1.3cm]
	    1 \arrow[r,"1" above] \arrow[d,"a" left,shift right=0.6ex]
          & 3 \arrow[r,"\pi" above] \arrow[d,"\q1" left,shift right=0.6ex]
	  & 5 \arrow[r,"1" above]  \arrow[d, "a" left, shift right=0.6ex]
	  & 7 \arrow[d,"\q1" left,shift right=0.6ex]
         \\ 2 \arrow[u,"\p1" right,shift right=0.6ex]
          & 4 \arrow[u,"b" right,shift right=0.7ex]
          & 6 \arrow[u,"\p1" right,shift right=0.6ex]
	  & 8 \arrow[u,"b" right,shift right=0.6ex]
	\end{tikzcd}
	\]
	\vskip-3mm
	\caption{Graph of $\tilde{A}$\label{fig-3}}
	\end{minipage}
	\end{figure}

	The automaton of $\psi$ is $\Z{\tilde{S},\tilde{A},\tilde{F}}$ and has 8 states, with initial state 1 coded by $\tilde{S}_{1,1} = 1$,
	accepting state 7 coded by $\tilde{F}_{7,1} = 1$, and transition matrix $\tilde{A} = \tilde{U} +
        (\tilde{X} + \pi W) + \tilde{V}$ shown in Figure \ref{fig-4}.
	\begin{figure}[ht]
	\begin{minipage}{15.5cm}
	\renewcommand{\arraycolsep}{4pt}
	\[\qquad
		\begin{pmatrix}
			0	& 0	& 0	& 0	& 0	& 0	& 0	& 0 \\
			\p1	& 0	& 0	& 0	& 0	& 0	& 0	& 0 \\
			0	& 0	& 0	& 0	& 0	& 0	& 0	& 0 \\
			0	& 0	& 0	& 0	& 0	& 0	& 0	& 0 \\
			0	& 0	& 0	& 0	& 0	& 0	& 0	& 0 \\
			0	& 0	& 0	& 0	& \p1	& 0	& 0	& 0 \\
			0	& 0	& 0	& 0	& 0	& 0	& 0	& 0 \\
			0	& 0	& 0	& 0	& 0	& 0	& 0	& 0 \\
		\end{pmatrix} + \begin{pmatrix}
			0	& a	& 1	& 0	& 0	& 0	& 0	& 0 \\
			0	& 0	& 0	& 0	& 0	& 0	& 0	& 0 \\
			0	& 0	& 0	& 0	& \pi	& 0	& 0	& 0 \\
			0	& 0	& b	& 0	& 0	& 0	& 0	& 0 \\
			0	& 0	& 0	& 0	& 0	& a	& 1	& 0 \\
			0	& 0	& 0	& 0	& 0	& 0	& 0	& 0 \\
			0	& 0	& 0	& 0	& 0	& 0	& 0	& 0 \\
			0	& 0	& 0	& 0	& 0	& 0	& b	& 0 \\
		\end{pmatrix} + \begin{pmatrix}
			0	& 0	& 0	& 0	& 0	& 0	& 0	& 0 \\
			0	& 0	& 0	& 0	& 0	& 0	& 0	& 0 \\
			0	& 0	& 0	& \q1	& 0	& 0	& 0	& 0 \\
			0	& 0	& 0	& 0	& 0	& 0	& 0	& 0 \\
			0	& 0	& 0	& 0	& 0	& 0	& 0	& 0 \\
			0	& 0	& 0	& 0	& 0	& 0	& 0	& 0 \\
			0	& 0	& 0	& 0	& 0	& 0	& 0	& \q1 \\
			0	& 0	& 0	& 0	& 0	& 0	& 0	& 0 \\
		\end{pmatrix}.
	\]
	\vskip-3mm
	\caption{Transition matrix $\tilde{A} = \tilde{U}+(\tilde{X}+\pi W) + \tilde{V}$\label{fig-4}}
	\end{minipage}
	\end{figure}

	For $N = \mu y.(\tilde{U}y\tilde{V}+\tilde X)^* = \p0(\tilde{U}\p1 + \tilde{X} + \q1\tilde{V})^*\q0 \in (\zc{\CR X^*}{C_2'})^{8\times 8}$, the proof shows that
	\begin{eqnarray*}
		(\tilde U+\tilde X+\tilde V)^*		&=& (N\tilde V)^* N (\tilde U N)^*, \\
		\p0 (\tilde U+\tilde X+\tilde V)^* \q0	&=& N, \\
		\p0 (\tilde A)^* \q0			&=& N (WN)^*.
	\end{eqnarray*}
	Since the graph of $(\tilde{U}\p1+\tilde{X}+\q1\tilde{V})$ consists of two disconnected components isomorphic to that of $A$ of Figure \ref{fig-2} above,
	the matrix $N$ obtained from its transitive reflexive hull
	is, using $\wL=\p0(a\p1^2)^*(\q1^2b)^*\q0$ as in Example~\ref{expl-3},
	\begin{eqnarray*}
	  N\ =\ \p0(\tilde{U}\p1+\tilde{X}+\q1\tilde{V})^*\q0
           &=& \begin{pmatrix}
			1	& a	& \wL	& a\wL	& 0	& 0	& 0	& 0 \\
			0	& 1	& \wL b	& \wL	& 0	& 0	& 0	& 0 \\
			0	& 0	& 1	& 0	& 0	& 0	& 0	& 0 \\
			0	& 0	& b	& 1	& 0	& 0	& 0	& 0 \\
			0	& 0	& 0	& 0	& 1	& a	& \wL	& a\wL \\
			0	& 0	& 0	& 0	& 0	& 1	& \wL b	& \wL \\
			0	& 0	& 0	& 0	& 0	& 0	& 1	& 0 \\
			0	& 0	& 0	& 0	& 0	& 0	& b	& 1
		\end{pmatrix}.
	\end{eqnarray*}

	The matrix $W$ is the boolean $8\times 8$-matrix with 1 only at $W_{3,5}$, so
	\begin{eqnarray*}
		\lefteqn{\p0(\tilde{A}^*)\q0 = N (WN)^*} \\
	&=&	N \begin{pmatrix}
		  0	& 0	& 0	& 0	& 0	& 0	& 0	& 0 \\
		  0	& 0	& 0	& 0	& 0	& 0	& 0	& 0 \\
		  0	& 0	& 0	& 0	& 1	& a	& \wL	& a\wL \\
		  0	& 0	& 0	& 0	& 0	& 0	& 0	& 0 \\
		  0	& 0	& 0	& 0	& 0	& 0	& 0	& 0 \\
		  0	& 0	& 0	& 0	& 0	& 0	& 0	& 0 \\
		  0	& 0	& 0	& 0	& 0	& 0	& 0	& 0 \\
		  0	& 0	& 0	& 0	& 0	& 0	& 0	& 0 \\
		\end{pmatrix}^* = N\begin{pmatrix}
		1	& 0	& 0	& 0	& 0	& 0	& 0	& 0 \\
		0	& 1	& 0	& 0	& 0	& 0	& 0	& 0 \\
		0	& 0	& 1	& 0	& 1	& a	& \wL	& a\wL \\
		0	& 0	& 0	& 1	& 0	& 0	& 0	& 0 \\
		0	& 0	& 0	& 0	& 1	& 0	& 0	& 0 \\
		0	& 0	& 0	& 0	& 0	& 1	& 0	& 0 \\
		0	& 0	& 0	& 0	& 0	& 0	& 1	& 0 \\
		0	& 0	& 0	& 0	& 0	& 0	& 0	& 1 \\
		\end{pmatrix} 
        \end{eqnarray*}
        \begin{eqnarray*}
		&=&	\begin{pmatrix}
		  1	& a	& \wL	& a\wL	& \wL	& a\wL		& \wL\wL	& \wL a\wL \\
		  0	& 1	& \wL b	& \wL	& \wL b	& \wL ba	& \wL b\wL	& \wL ba\wL \\
		  0	& 0	& 1	& 0	& 1	& a		& \wL		& a\wL \\
		  0	& 0	& b	& 1	& b	& ba		& b\wL		& ba\wL \\
		  0	& 0	& 0	& 0	& 1	& a		& \wL		& a\wL \\
		  0	& 0	& 0	& 0	& 0	& 1		& \wL b		& \wL \\
		  0	& 0	& 0	& 0	& 0	& 0		& 1		& 0 \\
		  0	& 0	& 0	& 0	& 0	& 0		& b		& 1 \\
		\end{pmatrix}.
	\end{eqnarray*}
	Hence $\p0\psi\q0 = \p0\tilde{S}\tilde{A}^*\tilde{F}\q0 = \p0(\tilde{A}^*)_{1,7}\q0 = (N(WN)^*)_{1,7} = \wL\wL$.
        As shown in Example \ref{expl-1},
	\[
		\wL= \p0(a\p1^2)^*(\q1^2b)^*\q0 =\p0(a\p1)^*(\q1b)^*\q0 = \p0\phi\q0 = \sum L \in Z_{C_2'} K,
	\]
	and since $Z_{C_2'} K$ is a $\CC$-dioid by Theorem \ref{thm-centralizer} \ref{item-4}, $\wL\wL = (\sum L)(\sum L) = \sum (LL)$.
	So, $\p0\psi\q0 = \p0\phi\q0\p0\phi\q0$ represents $LL$ in $K$.
	This can also be seen using \Star-continuity as in Example \ref{expl-1}.
\qee
\end{example}

\section{Combining normal forms by Kleene algebra operations}
\label{sec-combining-nfs}

We here show that normal forms for elements of $K\xR C_2'$ can be defined directly by
induction on the regular operations, using the representation by automata
only implicitly.

\begin{theorem}\label{thm-combining-nfs}
  Let $K$ be an $\CR$-dioid. For every $\phi \in K\xR C_2'$ there are $n\geq1$, $S\in
  \B^{1\times n}$, $F\in \B^{n\times 1}$, $U\in \{0,\p0,\p1\}^{n\times n}$, $V\in
  \{0,\q0,\q1\}^{n\times n}$, and $N\in (\zc{K}{C_2'})^{n\times n}$ such that
  \[ \phi = S(NV)^*N(UN)^*F. \]
Moreover, $N$ is the least solution of $y\geq (UyV + N)^*$ in $\Mat nn{\zc{K}{C_2'}}$.
\end{theorem}

\begin{proof}
  For the second claim, we only show $N\geq (UNV + N)^*$, since any $y\geq (UyV+N)^*$
  is above $N$. The first claim is shown by induction on $\phi$, choosing $S,U,N,V,F$
  from an implicit automaton $\Z{S,A,F}$ for $\phi=SA^*F$ (cf.~Theorem \ref{thm-aut})
  with $A=U+X+V$ and $N=\mu y.(UyV+X)^*$.
  \begin{description}
    \item[$\phi \in \{0,1\}$:] Put $n=1$, $U=V=(0)$ and $F=N=(1)$. Then $(UNV+N)^*
      \leq N$ and
      \[ (NV)^*N(UN)^* = (0)^*N(0)^* = N = (1). \]
      We have $S(NV)^*N(UN)^*F = \phi$ if we take $S=(0)$ for $\phi=0$ and $S=(1)$
      for $\phi=1$.
  \end{description}

In the remaining cases, this instance of the recursion formula
(\ref{mat-star}) for matrix iteration is used often:
\begin{equation}\label{mat-star-instance}
	\begin{pmatrix}
		A	& B \\
		0	& D
	\end{pmatrix}^* = \begin{pmatrix}
		A^*	& A^*B D^*& \\
		0	& D^*
	\end{pmatrix}.
\end{equation}

For the remaining generators, i.e.~the images in $K\xR C_2'$ of $k\in K$ or
$\p{0},\p{1},\q{0},\q{1}\in C_2'$, let $n=2$ and $S,U,N,V,F$ as
shown below.
\begin{description}
  \item[$\phi = k\in K$:] Here,
\begin{eqnarray*}
  \lefteqn{S(NV)^*N(UN)^*F} \\
  &=& \Sk \left(\Nk\nullmatrix\right)^*\Nk \left(\nullmatrix\Nk\right)^* \Fk
  \\ &=& \Sk \Nk \Fk = k = \phi.
\end{eqnarray*}
By (\ref{mat-star-instance}), $(UNV+N)^* = N^* = N$.

\item[$\phi \in\{\p{i},\q{i}\}$:] If $\phi$ is an opening bracket $\p{i}$,
  or, respectively, a closing bracket $\q{i}$, let $U,N,V$ be
	\[
		\begin{pmatrix}
			0	& \p{i} \\
			0	& 0
		\end{pmatrix}, \begin{pmatrix}
			1	& 0 \\
			0	& 1
		\end{pmatrix}, \begin{pmatrix}
			0	& 0 \\
			0	& 0
		\end{pmatrix},
	\quad\textrm{respectively}\quad
		\begin{pmatrix}
			0	& 0 \\
			0	& 0
		\end{pmatrix}, \begin{pmatrix}
			1	& 0 \\
			0	& 1
		\end{pmatrix}, \begin{pmatrix}
			0	& \q{i} \\
			0	& 0
		\end{pmatrix}.
	        \]
Then, using $S$ and $F$ as for $\phi=k$ above, $S(NV)^*N(UN)^*F = SU^*F =
SUF = \p{i}$ and, respectively, $S(NV)^*N(UN)^*F = SV^*F = SVF = \q{i}$.
\end{description}

For $\phi$ among $\phi_1+\phi_2$, $\phi_1\cdot \phi_2$, and $\phi_1^+$,
suppose that for $i = 1, 2$, by induction we have $n_i\geq 1$ and
$S_i,U_i,V_i,F_i$ and $N_i$ such that 
$\phi_i = S_i(N_iV_i)^* N_i (U_iN_i)^*F_i$ and $N_i = \mu y.(U_iyV_i+N_i)^*$.

\begin{description}
\item[$\phi = \phi_1+\phi_2$:] Let $n=n_1+n_2$ and $S,U,N,V,F$ as shown in
\begin{eqnarray*}
  \lefteqn{S(NV)^*N(UN)^*F}
  \\ &=& \Sadd\left(\Nadd\Vadd\right)^*\Nadd\left(\Uadd\Nadd\right)^*\Fadd
  \\ &=& \Sadd
         \begin{pmatrix} (N_1V_1)^* & 0 \\ 0 & (N_2V_2)^* \end{pmatrix}
         \Nadd
         \begin{pmatrix} (U_1N_1)^* & 0 \\ 0 & (U_2N_2)^* \end{pmatrix}
         \Fadd
  \\ &=& \Sadd
         \begin{pmatrix} (N_1V_1)^*N_1(U_1N_1)^* & 0 \\
                              0 & (N_2V_2)^*N_2(U_2N_2)^* \end{pmatrix}
         \Fadd
 \\ &=& S_1(N_1V_1)^*N_1(U_1N_1)^*F_1 + S_2 (N_2V_2)^*N_2(U_2N_2)^*F_2
 \\ &=& \phi_1+\phi_2.
\end{eqnarray*}
For the second claim, from $(U_iN_iV_i+N_i)^*\leq N_i$ we obtain
\[ (UNV + N)^* =
  \begin{pmatrix}U_1N_1V_1 + N_1 & 0 \\ 0 & U_2N_2V_2 + N_2\end{pmatrix}^*
  \leq \begin{pmatrix}N_1 & 0 \\ 0 & N_2\end{pmatrix} = N.
\]

\item[$\phi = \phi_1\cdot\phi_2$:] Notice that entries of
  $N_1F_1S_2N_2$ belong to the centralizer, and if $z$ is an $n_1\times n_2$ matrix
  of elements $x$ of the centralizer, so is $U_1zV_2$, because its entries are 0 or
  sums of elements $\p{i} x \q{j} = x\cdot\delta_{i,j}$, which belong to the
  centralizer.  Hence, $f(z) = N_1U_1zV_2N_2 + N_1F_1S_2N_2$ defines a monotone map
\[ f:\mat{n_1}{n_2}{(\zc{K}{C_2'})} \to \mat{n_1}{n_2}{(\zc{K}{C_2'})}, \]
By Theorem \ref{thm-centralizer} \ref{item-4}, $\zc{K}{C_2'}$ is a $\CC$-dioid, hence
a Chomsky algebra, so $f$ has a least pre-fixpoint $\alpha$, i.e.~the system
of $n_1 n_2$ polynomial inequations
\begin{equation}\label{eqn-alpha}
  z \geq 
  N_1U_1zV_2N_2 + N_1F_1S_2N_2
\end{equation}
has $\alpha$ as least solution\footnote{~The entries of $\alpha$ can be given
  as regular $\mu$-terms from the polynomials of (\ref{eqn-alpha}), as shown in
  \cite{Leiss2016}, or as regular expressions in the parameters of (\ref{eqn-alpha})
  and the brackets of $C_2'$, by a method presented in Theorem 15 and Example 6 of
  \cite{Leiss22}.  Alternatively, by Lemma \ref{lem-N-matrix-C2'}, using $A =
  U_1p+X_1+qV_1$ and $D = U_2p+X_2+qV_2$ with automaton $\Z{S_i,U_i+X_i+V_i,F_i}$ for
  $\phi_i$,
  \[
    N = b(Up+X+qV)^*d
    =  b\begin{pmatrix} A & F_1S_2 \\ 0 & D \end{pmatrix}^*d
    =  b\begin{pmatrix} A^* & A^*F_1S_2D^* \\ 0 & D^* \end{pmatrix}d
    =  \begin{pmatrix}N_1 & bA^*F_1S_2D^*d \\ 0 & N_2 \end{pmatrix},
  \]
so we also have $\alpha = bA^*F_1S_2D^*d$, but it is not obvious that this is
a matrix of elements from the centralizer.}.
Let $n=n_1+n_2$ and $S,U,N,V,F$ as in
\begin{eqnarray*}
  \lefteqn{S(NV)^*N(UN)^*F} \\
  &=& \Smult\left(\Nmult\Vadd\right)^*\Nmult\left(\Uadd\Nmult\right)^*\Fmult\\
\end{eqnarray*}
\begin{eqnarray*}
&=& \Smult
      \begin{pmatrix} N_1V_1 & \alpha V_2 \\ 0 & N_2V_2 \end{pmatrix}^*
      \Nmult
      \begin{pmatrix} U_1N_1 & U_1\alpha \\ 0 & U_2N_2 \end{pmatrix}^*
      \Fmult \\
  &=& \Smult
      \begin{pmatrix} (N_1V_1)^* & (N_1V_1)^*\alpha V_2 (N_2V_2)^*\\
                               0 & (N_2V_2)^* \end{pmatrix}
      \Nmult \\ &&
      \hskip37mm\begin{pmatrix} (U_1N_1)^* & (U_1N_1)^*U_1\alpha(U_2N_2)^* \\
                               0 & (U_2N_2)^* \end{pmatrix}
      \Fmult \\
  &=& \begin{pmatrix} S_1(N_1V_1)^* & S_1(N_1V_1)^*\alpha V_2 (N_2V_2)^*\end{pmatrix}
      \Nmult
      \begin{pmatrix} (U_1N_1)^*U_1\alpha(U_2N_2)^*F_2 \\
                          (U_2N_2)^*F_2 \end{pmatrix} \\
  &=& S_1(N_1V_1)^*[N_1(U_1N_1)^*U_1\alpha + \alpha + \alpha V_2 (N_2V_2)^*N_2](U_2N_2)^*F_2 .
\end{eqnarray*}
By \cite{Leiss2016}, the $\mu$-continuity of $\zc{K}{C_2'}$ lifts to the
matrix level, so
\begin{eqnarray*}
  \alpha &=& \Sum{(N_1U_1)^k(N_1F_1S_2N_2)(V_2N_2)^k}{k\in\N}
  \\ &=& \Sum{N_1(U_1N_1)^kF_1S_2(N_2V_2)^kN_2}{k\in\N}
\end{eqnarray*}
and
\begin{eqnarray*}
  \lefteqn{S(NV)^*N(UN)^*F}
  \\ &=& S_1(N_1V_1)^*[N_1(U_1N_1)^*U_1\alpha + \alpha + \alpha V_2(N_2V_2)^*N_2](U_2N_2)^*F_2
  \\ &=& 
  \Sum{S_1(N_1V_1)^*N_1(U_1N_1)^*U_1N_1(U_1N_1)^kF_1S_2(N_2V_2)^kN_2(U_2N_2)^*F_2}{k\in\N}
  \\ & & + \Sum{S_1(N_1V_1)^*N_1(U_1N_1)^kF_1S_2(N_2V_2)^kN_2(U_2N_2)^*F_2}{k\in\N}
  \\ & & + \Sum{S_1(N_1V_1)^*N_1(U_1N_1)^kF_1S_2(N_2V_2)^kN_2V_2(N_2V_2)^*N_2(U_2N_2)^*F_2}{k\in\N}
  \\ &=& \Sum{S_1(N_1V_1)^*N_1(U_1N_1)^kF_1S_2(N_2V_2)^lN_2(U_2N_2)^*F_2}{k,l\in\N}
  \\ &=& S_1(N_1V_1)^*N_1(U_1N_1)^*F_1\cdot S_2(N_2V_2)^*N_2(U_2N_2)^*F_2
  \\ &=& \phi_1\cdot \phi_2.
\end{eqnarray*}

For the second claim,
\[ (UNV + N)^* =
\begin{pmatrix}U_1N_1V_1 + N_1 & U_1\alpha V_2 + \alpha \\ 0 & U_2N_2V_2 +
  N_2\end{pmatrix}^*
  =
\begin{pmatrix}N_1 & N_1(U_1\alpha V_2 + \alpha)N_2 \\ 0 & N_2\end{pmatrix}.
\]
  Since $N_iN_i\leq N_i^*\leq N_i$, we have $N_1\alpha N_2 \leq \alpha =
  N_1U_1\alpha V_2N_2$, hence $(UNV+N)^*\leq N$.

\item[$\phi = \phi_1^+$:] Let $n,S,U,V,F$ be $n_1,S_1,U_1,V_1,F_1$. Since $N_1$ is
  the least solution of $y\geq (UyV+N_1)^*$, by Theorem \ref{thm-Nxuv}
  $(N_1V)^*N_1(UN_1)^*=A_1^*$ for $A_1 = U+N_1+V$, so $\phi_1=SA_1^*F$.  Then
\[ S(A_1+FS)^*F = SA_1^*(FSA_1^*)^*F = SA_1^*F(SA_1^*F)^* = \phi_1\phi_1^* = \phi_1^+.
	\]
By Lemma \ref{lem-N-matrix-C2'}, $y\geq (UyV + N_1+FS)^*$ has a least solution $N$, and
$N\in\mat nn{(\zc{K}{C_2'})}$. Using Theorem \ref{thm-Nxuv} again,
\[ S(NV)^*N(UN)^*F = S(U+N_1+FS+V)^*F = S(A_1+FS)^*F = \phi^+. \]
For the second claim, by definition of $N$ we have $N = (UNV + N_1+FS)^*$, so
\begin{eqnarray*}
  (UNV+N)^* &\leq& (UNV + (UNV+N_1+FS)^*)^*
  \\ &=& (UNV + UNV+N_1+FS)^* \leq N.
\end{eqnarray*}
\end{description}
  The case $\phi_1^*$ is treated via $\phi_1^* = 1 + \phi_1^+$.
\end{proof}

\section{Bra-ket $\CR$-dioids $C_m$ and the completeness property} \label{sec-bra-ket}

The \blue{bra-ket $\CR$-dioid $C_m$} is the quotient $\CR\Delta_m^*/\rho_m$ of $\CR\Delta_m^*$ by the $\CR$-congruence $\rho_m$ generated by the relations
\[
	\setof{p_iq_j=\delta_{i,j}}{i,j<m} \cup \{q_0p_0+ \ldots + q_{m-1}p_{m-1} = 1\}.
\]
While the match equations can be interpreted in monoids with an annihilating element
$0$, such as the polycyclic monoid $P_m'$, the completeness equation $1=\sum_{i <
  m}q_ip_i$ is a semiring equation. 

The name \emph{bra-ket} $\CR$-dioid comes by analogy to a notation in quantum
mechanics, where a quantum state is represented by a vector $\psi$ of a Hilbert space
$\CH$, written $\br\psi$ and called a \emph{ket}. Elements $f$ of the dual space $\CH^*$
of (continuous) linear functions on $\CH$ can uniquely be represented by elements
$\phi\in \CH$, via $f(\psi) = \Z{\phi,\psi}$ for all $\psi\in \CH$, where
$\Z{\cdot,\cdot}:\CH\times \CH \to F$ is the inner product on $\CH$ to the underlying
field $F$. The element of $\CH^*$ represented by $\phi$ is written $\bl\phi$ and called
a \emph{bra}.

Suppose $\CH$ has finite dimenson $m$ and let $\br0,\ldots,\br{m-1}$ be a basis of $\CH$ of
unit column vectors and $\bl0,\ldots,\bl{m-1}$ a basis of $\CH^*$ of
unit row vectors. If the application of $\bl{i}=(i_0,\ldots, i_{m-1})\in \CH^*$
to the vector $\br{j}$ with row values $j_0=\delta_{0,j}\ldots,j_{m-1}=\delta_{m-1,j}$
is written as juxtaposition, we get the bracket match- and mismatch equation for the inner product,
\[ \bl{i} \br{j} = \Z{i,j} = \Sum{i_kj_k}{k<m} = \delta_{i,j}. \]
The outer product $\br j\bl i$ of
$\q j$ and $\p i$ is a linear operator on $\CH$ and represented by the $m\times m$ matrix
\[ \br j \bl i = \big(j_ki_l\big).
\]
In particular, $\br i\bl i$ is a projection to the subspace spanned by $\br i$ and
represented by the $m\times m$-matrix with $1$ on the $i$-th position on the diagonal
and 0 otherwise. The combination of the projections gives the identity operator
\[ \br0\bl0 + \ldots + \br{m-1}\bl{m-1} = 1, \]
represented by the unit matrix of dimension $m$, corresponding to the completeness
equation.  This interpretation of $\bl{i}\br{j}$ and $\br j\bl i$ has to be combined
with an interpretation of $\br i\br j$ as a tensor in the 2-particle space
$\CH\otimes\CH$. Here, opening and closing brackets are interpreted by different
kinds of objects and strings of brackets are interpreted in several ways.  A uniform
interpretation of brackets and bracket concatenation is given below.

\subsection{The bra-ket $\CR$-dioid $C_m$ and matrix algebras}\label{sec-mat-Cm}

For applications to context-free languages $L\subseteq X^*$, the $\CR$-dioids $C_m$
and $C_m'$ as factors $C$ of $\CR X^*\xR C$ arise by an interpretation of brackets
$p_i$ as pushing and $q_i$ as popping symbol $i$ from a stack. Then $p_iq_i$ leaves
the stack unchanged, $p_iq_j$ for $j\not=i$ aborts the computation, and $q_ip_i$
succeeds iff $i$ is on top of the stack. The completeness equation $\sum_{i<m}
q_ip_i = 1$ of $C_m$ says that one of the symbols $i<m$ is always on top of the
stack (including 0 as end marker). This originally seemed necessary to have
$\zc{\CR X^*}{C}$ be isomorphic to the $\CC$-dioid of context-free languages over
$X^*$, but as shown in \cite{Leiss22}, $C=C_m'$ is sufficient.

More precisely, a uniform interpretation of brackets as binary relations
on a countably infinite set and bracket concatenation as relation product, i.e. an
interpretation of $C_m$ in $\Mat\omega\omega\B$, is as follows:

\begin{example}\label{expl-Mat-omega}
  \newcommand\e[1]{e_{#1}} Let $\e{0}, \e{1}, \ldots$ be the unit vectors of size
  $\omega\times 1$ and $\e{0}^t, \e{1}^t,\ldots$ their transposed vectors of size
  $1\times \omega$.  Each $\e{k}\e{l}^t$ is a boolean square matrix of dimension
  $\omega$, representing the relation $\{(k,l)\}\subseteq \omega\times\omega$, and
  so, for $i,j<m$, we can interpret brackets $\p{i}$ and $\q{j}$ in $\mat\omega\omega\B$ by
\begin{eqnarray*}
\p{i} = \sum_{k<\omega} \e{k} \e{mk + i}^t, &\quad&
\q{j} = \sum_{k<\omega} \e{mk + j}\e{k}^t,
\end{eqnarray*}
representing the relations $\setof{(k,mk+i)}{k\in\N}$ and $\setof{(mk+j,k)}{k\in\N}$, respectively.
Concatenation of brackets is boolean matrix multiplication, corresponding to relation
composition, so 
\[ \p{i}\q{j} = \delta_{i,j} \]
holds, with 0 and 1 for the zero and unit square matrices of dimension $\omega$. The matrix
$\q{i}\p{i}$ represents the subrelation $\setof{(mk+i,mk+i)}{k\in\N}$ of the
identity, so the completeness equation 
\[ \sum_{i<m} \q{i}\p{i} = 1 \]
also holds.  In a similar spirit, one can think of $\Gamma = \{0,\ldots,m-1\}$ as a
stack alphabet, $\Gamma^*$ as the set of possible stack contents $w$ (with top of the
stack on the left), and let $\p{i}$ be
the graph of the operation ``push symbol $i$'', $\q{i}$ the graph of ``pop symbol
$i$'', $\cdot$ the relation product, $+$ the union of relations, 0 the empty relation
and 1 the identity relation on $\Gamma^*$. Then clearly $\p{i}\q{j} = \delta_{i,j}$
holds, but, since one cannot pop from the empty stack $\epsilon$, $e:=
\sum_{i<m}\q{i}\p{i}$ is the identity relation on the non-empty stack $\Gamma^+$
only, so $e < 1 = e + \{(\epsilon,\epsilon)\}$. To obtain $e=1$, one can treat 0 as a
special symbol, pad all stack contents $w\in \Gamma^*$ by an $\omega$-sequence of 0's
to $w0^\omega$, and interpret the operations as binary relations on the new stack
$\Gamma^*0^\omega$. 
\qee\end{example}

A remarkable consequence of the completeness equation is the following:

\begin{theorem}\label{thm-matrix-C}
  $C_m$ is isomorphic to its own matrix Kleene algebra $\Mat mm{C_m}$.
\end{theorem}

\begin{proof}
  \setcounter{claim}0
  Define $\hat\cdot : C_m\to \Mat mm{C_m}$ and $\check\cdot:\Mat mm{C_m}\to C_m$ by
  \[ \hat a_{ij} := p_iaq_j \quad\textrm{for }a\in C_m, \qquad\textrm{and}\qquad
  \check A = \sum_{i,j<m}q_iA_{ij}p_j \quad\textrm{for }A\in\Mat mm{C_m}.
  \] 
  These maps are inverse to each other, because for $a\in C_m$ and $A\in \Mat mm{C_m}$,
  \[ \check{\hat a} = \sum_{i,j} q_i{\hat a_{ij}}p_j = \sum_{i,j}q_ip_iaq_jp_j
  = (\sum_i q_ip_i)a(\sum_j q_jp_j) = a,
  \]
  \[ ({\check A})\hat{}_{kl} = (\sum_{i,j}q_iA_{ij}p_j)\hat{}_{kl} =
  p_k(\sum_{i,j}q_iA_{ij}p_j)q_l = p_kq_kA_{kl}p_lq_l = A_{kl}.
 \]
  Let $0_m$ be the zero and $1_m$ the unit matrix of dimension $m\times m$.
  Clearly, $\hat\cdot$ is a semiring morphism, by
  \[ \hat 0 = \big(p_i0q_j\big) = 0_m, \]
  \[ \hat 1 = \big(p_i1q_j\big) = \big(\delta_{ij}\big) = 1_m,\]
\[ \hat a+\hat b = \big(p_iaq_j\big)+\big(p_ibq_j\big) = \big(p_i(a+b)q_j\big) = \widehat{a+b},\]
\[ \hat a\cdot \hat b = \big(\sum_k p_iaq_kp_kbq_j\big) = \big(p_ia(\sum_kq_kp_k)bq_j\big)
= \big(p_iabq_j\big) = \widehat{\,ab\,}.
\]
We leave it to the reader to check that the inverse $\check{\,\cdot\,}$
also is a semiring morphism.  Since they preserve $+$, these maps are
monotone and order isomorphisms. To see that they are Kleene algebra
morphisms, let $a\in C_m$ and $A\in \Mat mm{C_m}$. Then $a^*=\mu x.g_a(x)$ and $A^*=\mu
x.h_A(x)$ are the least pre-fixpoints of the monotone maps $g_a:C_m\to C_m$
and $h_A:\Mat mm{C_m}\to \Mat mm{C_m}$ defined by $g_a(x) = ax+1$ and $h_A(x) = Ax +
1_m$. For $f=\widehat\cdot:C_m\to \Mat mm{C_m}$ we have\[ (f\circ
g_a)(x) =\widehat{ax+1} = \hat a \hat x + \hat 1 = (h_{\hat a}\circ
f)(x). \] It follows that
\[ \widehat{a^*} = f(\mu x.g_a(x)) = \mu x.h_{\hat a}(f(x)) = \hat{a}^*. \]
  Likewise, for the inverse $f^{-1} = \check{\cdot}:\Mat mm{C_m}\to C_m$ we have
  \[ (f^{-1}\circ h_A)(x) = (Ax+1_m)\check{} = \check{A}\check{x} + 1 =
  (g_{\check A}\circ f^{-1})(x), \]
  which implies $(A^*)\check{} = \check{A}^*$.
\end{proof}

\begin{corollary}
The Kleene subalgebra of $C_m$ generated by $\setof{q_ip_j}{i,j<m}$ is isomorphic to $\Mat mm\B$.
Moreover, $C_m \simeq C_m\xR \Mat mm{\B}$.
\end{corollary}

\begin{proof}
  Let $E_{(i,j)}$ be the $m\times m$ boolean matrix with 1 only at position $(i,j)$.
  The first claim holds since the isomorphism $\check{\,\cdot\,}:\Mat mm{C_m} \to
  C_m$ maps a generator $E_{(i,j)}$ of $\Mat mm\B$ to $q_ip_j$.
  The second claim follows from $C_m\simeq \Mat mm{C_m}$ and Proposition \ref{prop-MatK}.
\end{proof}

Similar to Lemma \ref{lem-recoding-C'}, we can code $C_m$ in $C_2$ for $m>2$:

\begin{proposition}\label{lem-recoding-C}
  For $m > 2$ there is an embedding $\CR$-morphism $g: C_m \to C_2$
  such that for $i, j < m$,
  \[ g(p_i)\cdot g(q_j) = \delta_{i,j} \quad\textrm{and}\quad g(\sum_{i<m} q_ip_i) = 1, \]
  writing $p_i$ and $q_j$ for the congruence classes $\{p_i\}/\rho_m$ and $\{q_j\}/\rho_m$ in $C_m$.
\end{proposition}
\begin{proof}
  Let $\rho_m$ be the $\CR$-congruence on $\CR\Delta_m^*$ generated by the match- and
  completeness equations for $\CR\Delta_m$ and $\rho_2$ the corresponding
  $\CR$-congruence on $\CR\Delta_2^*$. Writing again $\Delta_2=\{b,p,d,q\}$, we
  modify the coding $\ol{\,\cdot\,}$ of $\Delta_m$ in $\Delta_2^*$ of Lemma
  \ref{lem-recoding-C'} by putting
\[ \ol{p_i} = \begin{cases} bp^i, & i < m-1 \cr p^i, & i= m-1\end{cases}
  \quad\textrm{and}\quad
   \ol{q_i} = \begin{cases} q^id, & i < m-1 \cr q^i, & i= m-1 . \end{cases}
   \]
This extends to a homomorphism from $\Delta_m^*$ to $\Delta_2^*$ and lifts to an
$\CR$-morphism $\ol{\,\cdot\,}: \CR\Delta_m^* \to \CR\Delta_2^*$. Clearly, the match
equations $\ol{p_i}\,\ol{q_j} =\delta_{i,j}$ for $i,j<m$ hold in
$C_2=\CR\Delta_2^*/\rho_2$. In $C_2$, we have
$1 = db + qp = \ol{q_0}\ol{p_0} + q^1p^1$, and since  for $1\leq i<m-1$
\[ q^ip^i = q^i(db+qp)p^i = q^idbp^i + q^iqpp^i = \ol{q_i}\ol{p_i}+q^{i+1}p^{i+1}, \]
it follows that
\[ 1 = \ol{q_0}\ol{p_0}+q^1p^1
     = \sum_{i<m-1} \ol{q_i}\ol{p_i} + q^{m-1}p^{m-1} = \sum_{i<m} \ol{q_i}\ol{p_i}. \]
So the completeness equation of $C_m$ also holds under the coding in $C_2$.  Hence a
map $g:\CR\Delta_m^*/\rho_m\to \CR\Delta_2^*/\rho_2$ is well-defined by $g(A/\rho_m)
= \ol{A}/\rho_2$ for $A\in\CR\Delta_m^*$. As in Lemma \ref{lem-recoding-C'}, it is an
$\CR$-morphism and satisfies $g(p_i)\cdot g(q_j) = \delta_{i,j}$ and $1 =
g(\sum_{i<m} q_ip_i)$.  (But the additional property $p_0\cdot g(q_i) = 0 =
g(p_i)\cdot q_0$ of Lemma \ref{lem-recoding-C'} only holds for $i>0$, since
$p_0\ol{q_0} = bq^0d = 1 = bp^0d = \ol{p_0}q_0$ in $C_2$.)

To see that $g$ is injective, first notice that $\ol{\,\cdot\,}:\CR\Delta_m^*\to
\CR\Delta_2^*$ is injective: any $w\in\Delta_2^*$ in the image of $\ol{\,\cdot\,}$
can uniquely be parsed into a word of
$\{\ol{p_0},\ldots,\ol{p_{m-1}},\ol{q_0},\ldots,\ol{q_{m-1}}\}^*$, so there is a
unique $v\in\Delta_m^*$ with $w=\ol{v}$. It is therefore sufficient to show
  \begin{equation}\label{eqn-rho-2,n}
    \textrm{for all $A,B\in \CR\Delta_m^*$}(\ol A \r2n \ol B \Rightarrow A/\rho_m = B/\rho_m),
  \end{equation}
  where $\r2n$ is the $n$-th stage of the inductive definition of $\rho_2$. This
  is done by induction on $n$.
  If $\ol A \r20 \ol B$, either $\ol A = \ol B$, in which case $A=B$, or $\ol A
  \r20\ol B$ is a match equation or the completeness equation of $\rho_2$, in
  which case $(A,B)$ is the corresponding match or completeness equation of
  $\rho_m$, so $A/\rho_m=B/\rho_m$.
  If $\ol A \r2{n+1}\ol B$ is obtained by symmetry from $\ol B \r2n \ol A$ or
  by transitivity from $\ol A\r2n \ol C$ and $\ol C\r2n \ol B$, the claim
  follows from symmetry resp.~transitivity of $\rho_m$.

  If $\ol{A}\r2{n+1}\ol{B}$ is obtained from $\ol{A_1}\r2n \ol{B_1}$ and
  $\ol{A_2}\r2n\ol{B_2}$ by $\ol{A} = \ol{A_1}\,\ol{A_2}$ and $\ol{B}=
  \ol{B_1}\,\ol{B_2}$, then
  \[ A/\rho_m = (A_1A_2)/\rho_m = A_1/\rho_m A_2/\rho_m
              = B_1/\rho_m B_2/\rho_m = (B_1B_2)/\rho_m = B/\rho_m
  \]
  by induction. The argument is similar if $\ol{A} = \ol{A_1}\cup\ol{A_2}$ and
  $\ol{B}=\ol{B_1}\cup\ol{B_2}$.

  Suppose ${\bigcup U'}\r2{n+1}{\bigcup V'}$, where $U', V'\in
  \CR(\CR\Delta_2^*)$ contain only regular sets of words in the image of
  $\ol{\,\cdot\,}:\Delta_m^*\to \Delta_2^*$ and $(U'/\rho_2)^\da = (V'/\rho_2)^\da$
  in stage $n$. As these regular sets of words are also regular sets of words over
  $\ol{\Delta_m}$, there are $U,V\in\CR(\CR\Delta_m^*)$ such that $U' =
  \setof{\ol A}{A\in U}$, $V' =\setof{\ol B}{B\in V}$, and for each $A\in U$ there
  is $B\in V$ with $\ol A\cup\ol B\r2n \ol B$ and for each $B\in V$ there is
  $A\in U$ with $\ol B\cup \ol A \r2n \ol A$. By induction, $\ol{A\cup B} = \ol
  A\cup\ol B\r2n \ol B$ implies $A/\rho_m\leq B/\rho_m$ and $\ol B\cup \ol
  A\r2n\ol A$ implies $B/\rho_m\leq A/\rho_m$, so that $(U/\rho_m)^\da =
  (V/\rho_m)^\da$ and therefore $\bigcup U/\rho_m = \bigcup U/\rho_m$. Since $\bigcup
  U'=\ol{\bigcup U}$ and $\bigcup V'=\ol{\bigcup V}$, the claim is proven.
\end{proof}

\subsection{Relativizing the completeness property}\label{sec-completeness}
Let $m \geq 2$ and $e := \sum_{i < m}{\q{i}\p{i}}$.
For the tensor product $K\xR C_m$ of an $\CR$-dioid $K$ and $C_m$, the completeness equation $e=1$ can be used to show that every element of the centralizer of $C_m$ is the least upper bound of some context-free subset of $K$, i.e.~that 
\[ \sum : \CC K \to \zc{K}{C_m} \] is surjective.
Also, Lemma \ref{lem-N-matrix-C2'} is a bit easier to prove for $C_2$ than for
$C_2'$, as we can use $db\leq 1$ to prove $NN\leq N$ in Claim \ref{claim-Lemma-1d}.
We do not go into this here, but observe that in suitable contexts, $e=1$ in a sense holds in
the polycyclic algebras $C_m'$ as well. For example, as $p_ie = p_i$ for $p_i \in P_m$ and $eq_j = q_j$
for $q_j \in Q_m$, in $C_m'$ we have
\[
	\p0 e\q0 = \p0 \q0 = 1 = \p0 1\q0.
\]
This can be generalized to a relativized form of the completeness property.
Basically, for any regular expression $\phi(x)$ in an unknown $x$, elements of $K$ and brackets
of $C_m'$ other than $p_0,q_0$, the two elements $\phi(e), \phi(1) \in K\xR C_m'$ are
suprema of regular sets that differ only by elements from the centralizer $\zc{K}{C_m'}$ weighted by
factors from $\{q_1,\ldots,q_{m-1}\}^*\{p_1,\ldots,p_{m-1}\}^*\setminus\{1\}$, and these
reduce to 0 in the context $\p0\ldots\q0$ of a fresh pair of brackets.

\begin{theorem}[Relative Completeness]
\label{thm-completeness-under-<0-0>}
	Let $K$ be an $\CR$-dioid.
	For any $\phi(x) =\phi(\pi, \p1, \ldots, \p{m-1}, \q1, \ldots, \q{m-1}, x)
        \in (K\xR C_m')[x]$ in which $\p0$ and $\q0$ occur only in $\pi = \q0\p0$,
	\[
		\p0\phi(e)\q0 = \p0\phi(1)\q0 \quad\textrm{and}\quad 
                \p0\phi(1)\q0 \in\zc{K}{C_m'}.
	\]
\end{theorem}

\begin{proof}
  Let $\phi(x)= S(A+W\pi)^*F$ be given by an automaton $\Z{S, A+W\pi, F}$,
	where $A = U+X+V+Yx$,  $S, U, X, V, W, F$
        and $n$ are as in Theorem \ref{thm-nf2}, and $Y \in \{0, 1\}^{n \times n}$.
        Then
        \[ \p0 \phi(x) \q0 = \p0S(A+W\pi)^*F\q0 = S\p0(A+W\pi)^*\q0F. \]
        Recall that on the right and in the following, $\p{i}$ and $\q{j}$ are
        identified with corresponding diagonal matrices of dimension $n$.
        As in the proof of Theorem \ref{thm-nf2},
	\begin{eqnarray*}
	  \p0 (A+W\pi)^*\q0	&=&	\p0 A^*\q0 (W\p0 A^*\q0)^*.
	\end{eqnarray*}
	It is sufficient to show that $\p0A^*\q0$ does not depend on the choice of $x \in \{1, e\}$.
	Using $\alpha= U+X+V$,
	\begin{eqnarray*}
		\p0 A^*\q0 &=& \p0 (\alpha+Yx)^*\q0 = \p0\alpha^*(Yx\alpha^*)^*\q0.
	\end{eqnarray*}
	Let $N \in (\zc{K}{C_m})^{n\times n}$ be as in Theorem \ref{thm-nf2}, so that
        with \Star-continuity on the matrix level,
	\begin{eqnarray*}
		\alpha^*	&=&	(U+X+V)^* = (NV)^* N (UN)^* \\
				&=&	\Sum{(NV)^kN (NU)^l}{k,l\in\N}.
	\end{eqnarray*}
	There are $U_i, V_j \in \B^{n\times n}$ such that $U=\sum_{0 < i < m} U_i\p{i}$ and $V=\sum_{0 < j < m}\q{j}V_j$.
        As the $q_j$ and $p_i$ commute with $N$ and boolean matrices,
	\begin{eqnarray*}
		\lefteqn{(NV)^kN (UN)^l	=	(\sum_{0 < j < m}\q{j}NV_j)^kN (\sum_{0 < i < m}U_iN\p{i})^l} \\
		&=&	\sum_{0< j_1, \ldots, j_k, i_1, \ldots, i_l < m} {\q{j_1}\cdots \q{j_k}\p{i_l}\cdots \p{i_1} NV_{j_1}\cdots NV_{j_k}NU_{i_l}N\cdots U_{i_1}N}.
	\end{eqnarray*}
	Let $P = P_m\setminus\{\p0\}$ and $Q = Q_m\setminus\{\q0\}$.
        For $v = \q{j_1} \cdots \q{j_k} \in Q^*$ and $u = \p{i_l} \cdots \p{i_1} \in P^*$, put
	\[
		N_{vu} = NV_{j_1}\cdots NV_{j_m}NU_{i_l}N\cdots U_{i_1}N,
	\]
	so that
	\[
		\alpha^* =	\Sum{(NV)^kN (UN)^l}{k,l\in\N} = \Sum{vuN_{vu}}{u\in P^*,v\in Q^*}.
	\]
	By \Star-continuity, it follows that
	\begin{eqnarray*}
	\lefteqn{\p0 \alpha^*(Ye\alpha^*)^*\q0 } \\
		&=&	\Sum{\p0\alpha^*(Ye\alpha^*)^k\q0}{k\in\N} \\
		&=&	\Sum{\p0 v_0u_0N_{v_0u_0}\ldots Yev_ku_kN_{v_ku_k}\q0
                           }{k\in\N, u_0, \ldots, u_k \in P^*, v_0, \ldots, v_k \in Q^*} \\
		&=&	\Sum{\p0 v_0u_0\ldots ev_ku_k\q0 N_{v_0u_0}\ldots YN_{v_ku_k}
                           }{k\in\N, u_0, \ldots, u_k \in P^*, v_0, \ldots, v_k \in Q^*},
	\end{eqnarray*}
	where the final step holds since the $ev_{i+1}u_{i+1}$ commute with $N_{v_0u_0}Y\cdots N_{v_iu_i}Y$.

	To show $\p0\alpha^*(Ye\alpha^*)^*\q0 =\p0\alpha^*(Y\alpha^*)^*\q0$,
	it therefore is sufficient that $e$ can be replaced by 1 in the summands
        $\p0 v_0u_0\ldots ev_ku_k\q0 N_{v_0u_0}\ldots YN_{v_ku_k}$, i.e.~that
        \begin{equation}\label{eqn-p0wk}
          p_0v_0u_0ev_1u_1\ldots ev_ku_k = p_0v_0u_0v_1u_1\ldots v_ku_k.
        \end{equation}
        For $k=0$, equation (\ref{eqn-p0wk}) is obvious. For $0<k$, put $w_j =
        v_0u_0\ldots v_ju_j$ and, by induction, assume
        \[ p_0v_0u_0ev_1u_1\ldots ev_ju_j = p_0w_j \]
        for some $j<k$. Since $w_j\in Q^*P^*\cup\{0\}$, we distinguish three
        cases. If $w_j=1$, then $p_0w_je = p_0e = p_0 = p_0w_j$, so
        $p_0w_jev_{j+1}u_{j+1} = p_0w_jv_{j+1}u_{j+1} = p_0w_{j+1}$.  If $w_j\in
        Q^+$, then $p_0w_j = 0$, so $p_0w_jev_{j+1}u_{j+1} = p_0w_{j+1}$. If
        $w_j\in Q^*P^+\cup\{0\}$, then $w_je = w_j$, so $p_0w_jev_{j+1}u_{j+1} =
        p_0w_{j+1}$.  It follows that $p_0v_0u_0ev_1u_1\ldots ev_{j+1}u_{j+1} =
        p_0w_{j+1}$, and by induction, (\ref{eqn-p0wk}).  Thus we have shown
        $\p0\alpha^*(Ye\alpha^*)^*\q0 =\p0\alpha^*(Y\alpha^*)^*\q0$ and thereby
        $\p0\phi(e)\q0 = \p0\phi(1)\q0$.

        Moreover, since $\p0w_k\q0 \in \{0,1\}$ for all $k$,
        $\p0\alpha^*(Y\alpha^*)^*\q0$ is the least upper bound of a regular set
        of $n\times n$-matrices over $\zc{K}{C_m'}$.  It follows that, for
        $x=1$,
        \[ \p0A^*\q0 = \p0\alpha^*(Y\alpha^*)^*\q0 \in \mat nn{(\zc{K}{C_m'})}, \]
        and therefore $\p0\phi(1)\q0 = S\p0A^*\q0(W\p0A^*\q0)^*F \in \zc{K}{C_m'}$.
\end{proof}

It therefore seems that at least for applications to formal languages, where we can
use a special pair $\p{0},\q{0}$ of brackets to annihilate words of $\{\q
1,\ldots,\q{m-1}\}^*\{\p 1,\ldots,\p {m-1}\}^*$, the completeness equation is of
little help.

\section{Conclusion}

The tensor product $\CR X^*\xR C_m'$ of the algebra $\CR X^*$ of regular sets of $X^*$ with the polycyclic Kleene algebra $C_m'$ based on $m \geq 2$ bracket pairs is a \Star-continuous Kleene algebra subsuming an isomorphic copy of the algebra $\CC X^*$ of context-free sets of $X^*$, the centralizer $\zc{\CR X^*}{C_m'}$ of $C_m'$.

We have investigated $K\xR C_m'$ for arbitrary \Star-continuous Kleene algebras $K$.
Every element $\phi \in K\xR C_m'$ is the value $SA^*F$ of an automaton $\Z{S, A, F}$ whose transition matrix $A=U+X+V$ splits into transitions
by opening brackets (and 0's) in $U$, transitions by elements of $K$ in $X$, and transitions by closing brackets (and 0's) in $V$.
Our main result is a normal form theorem saying that $A^*=(NV)^* N (UN)^*$,
where $N$ is the least solution of $y \geq (UyV+X)^*$ in $\Mat nn{K\xR C_m'}$,
corresponding to Dyck's language $D \subseteq \{U, X, V\}^*$ with bracket pair $U,V$, and $N$ has entries in $\zc{K}{C_m'}$.
If $\phi=SA^*F$ belongs to the centralizer of $C_2'$ in $K\xR C_2'$, and $K$ has no
zero divisors, then $SA^*F=SNF$. It remains open whether the non-existence of zero
divisors is a necessary assumption. These normal forms generalize a simpler normal form for
elements of the polycyclic monoid $P_m'[X]$.

Our main result had been obtained earlier (unpublished) by the first author with the
bra-ket Kleene algebra $C_m$ instead of $C_m'$.
For the brackets $p_0, \ldots, q_{m-1}$, we no longer need the completeness equation $1 = q_0p_0 + \ldots + q_{m-1}p_{m-1}$ of $C_m$,
but only the match- and mismatch equations $p_iq_j=\delta_{i,j}$ of $C_m'$.
It is also shown that in the context $p_0\ldots q_0$, in $K\xR C_m'$ this equation
can be assumed to hold.

The two sets of cases of greatest interest are specializations of $\CR M \xR C$ to
the monoids $M = X^*$ and $M = X^* \times Y^*$ and to $C = C_2'$ and $C =
C_2$. Applications, for $M = X^*$, include recognition of languages over an alphabet
of inputs $X$, while for $M = X^* \times Y^*$, they include parsing or translation of
languages over $X$, where $Y$ may denote an alphabet of actions (such as parse tree
building operations), or an alphabet of outputs. With the results established here,
we have laid a foundation for an algebraic study of recognition, parsing and
translation algorithms for context-free languages over $X$, that we hope to analyze in
greater depth in later publications.

In addition, given the close relation between $C_2'$ and $C_2$ and stack
machines, it is natural to enquire as to whether $\CR M \xR C$ may provide a
representation for 2-stack machine languages and relations, where $C = C_2' \xR
C_2'$ or $C = C_2 \xR C_2$, and, thus, a basis for a calculus for recursively
enumerable languages and relations.
We also hope to elaborate this in a future publication.

\section{Appendix}\label{sec-appendix}

We here complete the proof of Lemma \ref{lem-non-triv} by showing that
$\equiv\,\subseteq P$. We repeat that $P(R,S)$ is
  \begin{equation}\label{eqn-Q2}
    \forall (x,y),(a,b),(a',b')[ (a,b)R(a',b')\setminus Z \preceq (x,y)
      \iff  (a,b)S(a',b')\setminus Z \preceq (x,y)],
  \end{equation}
where $Z = \setof{(a,b)\in K_1\times K_2}{a=0 \textrm{ or }b=0}$ and $R\preceq (x,y)$
says that $(x,y)$ is an upper bound of $R\subseteq K_1\times K_2$.
\begin{proof}
  Let $\equiv_n$ be the $n$-th stage in the inductive definition of $\equiv$, where
  $\equiv_0$ consists of those $(R,S)$ where $R=S$ or where they are a tensor product
  equation, i.e.~$R= A\times B$ and $S=\{(\sum A,\sum B)\}$ for some $A\in \CR K_1,
  B\in\CR K_2$, and $\equiv_{n+1}$ adds pairs to $\equiv_n$ by the closure conditions
  for symmetry, transitivity, sum, product and supremum. To prove $\equiv\,\subseteq
  P$, it is sufficient to show $\equiv_n\,\subseteq P$ by induction on $n$.

Suppose $R\equiv_0 S$. If $R=S$, then $P(R,S)$ is clear, since $P$ is reflexive.
Otherwise, $R\equiv_0 S$ is a tensor product equation, i.e.~there are $A\in \CR K_1$
and $B\in \CR K_2$ such that $R=A\times B$ and $S=\{(\sum A,\sum B)\}$. Let
$(x,y),(a,b),(a',b')\in K_1\times K_2$. To show
\begin{equation}\label{eqn-AxB}
  (a,b)(A\times B)(a',b')\setminus Z \preceq (x,y) \iff
(a,b)\{(\sum A,\sum B)\}(a',b')\setminus Z\preceq (x,y),
\end{equation}
we first observe that, since for a rectangle
$A'\times B'\subseteq K_1\times K_2$,
\begin{eqnarray*}
  A'\times B'\subseteq Z &\iff& A'\subseteq\{0\} \vee B'\subseteq\{0\},
\end{eqnarray*}
either both sets $(a,b)(A\times B)(a',b')$ and $(a,b)\{(\sum A,\sum B)\}(a',b')$ are
subsets of $Z$ or both are not:
\begin{eqnarray*}
  (a,b)(A\times B)(a',b')\subseteq Z
  &\iff& aAa'\subseteq \{0\} \vee bBb'\subseteq\{0\} \\
  &\iff& \sum aAa' = 0 \vee \sum bBb'=0 \\
  &\iff& (a(\sum A)a',b(\sum B)b')\in Z
  \\ &\iff& (a,b)\{(\sum A,\sum B)\}(a',b')\subseteq Z.
\end{eqnarray*}
If both of these sets are subsets of $Z$, then clearly (\ref{eqn-AxB}) holds.
Otherwise, both $(a,b)(A\times B)(a',b')\setminus Z$ and $(a,b)\{(\sum A,\sum
B)\}(a',b')\setminus Z$ are non-empty. Since for rectangles $A'\times B'\not\subseteq
Z$,
\begin{eqnarray*}
  A'\times B' \setminus Z \preceq (x,y) &\iff& A'\times B'\preceq (x,y),
\end{eqnarray*}
the claim (\ref{eqn-AxB}) is implied by the following:
\begin{eqnarray*}
  (a,b)(A\times B)(a',b')\preceq (x,y)
  &\iff& (aAa'\times bBb') \preceq (x,y) \\
  &\iff& aAa' \preceq x \wedge bBb' \preceq y\\
  &\iff& \sum aAa' \leq x \wedge \sum bBb'\leq y \\
  &\iff& (a,b)\{(\sum A,\sum B)\}(a',b')\preceq (x,y).
\end{eqnarray*}

Suppose $R\equiv_{n+1}S$ is obtained from $S\equiv_n R$ by the condition to close
$\equiv$ under symmetry. By induction, $P(S,R)$ holds, and since $P$ is an
equivalence relation, $P(R,S)$ also holds.

Suppose $R\equiv_{n+1}S$ is obtained from $R\equiv_n T$ and $T\equiv_n S$ by the
condition to close $\equiv$ under transitivity. By induction, $P(R,T)$ and $P(T,S)$,
and since $P$ is an equivalence relation, $P(R,S)$.

Suppose $R_1\cup R_2\equiv_{n+1}S_1\cup S_2$ is obtained from $R_1\equiv_nS_1$ and
$R_2\equiv_nS_2$ by the condition to close $\equiv$ under union. By induction,
$P(R_1,S_1)$ and $P(R_2,S_2)$, and hence, for all $(a,b),(a',b')$ and $(x,y)$,
\begin{eqnarray*}
  \lefteqn{(a,b)(R_1\cup R_2)(a',b')\setminus Z \preceq (x,y)} \\
  &\iff& (a,b)R_1(a',b')\setminus Z \preceq (x,y) \wedge (a,b)R_2(a',b')\setminus Z \preceq (x,y) \\
  &\iff& (a,b)S_1(a',b')\setminus Z \preceq (x,y) \wedge (a,b)S_2(a',b')\setminus Z \preceq (x,y) \\
  &\iff& (a,b)(S_1\cup S_2)(a',b')\setminus Z \preceq (x,y),
\end{eqnarray*}
which shows $P(R_1\cup R_2, S_1\cup S_2)$.

Suppose $R_1R_2\equiv_{n+1}S_1S_2$ is obtained from $R_1\equiv_nS_1$ and
$R_2\equiv_nS_2$ by the condition to close $\equiv$ under products. Let
$(a,b),(a',b'),(x,y)\in K_1\times K_2$ and assume $(a,b)R_1R_2(a',b')\setminus Z
\preceq (x,y)$. By induction, $P(R_1,S_1)$, and hence, exploiting the universal
quantification in (\ref{eqn-Q2}),
\[ (a,b)S_1R_2(a',b')\setminus Z \preceq (x,y). \]
Since, by induction, we also have $P(R_2,S_2)$, this similarly gives
$(a,b)S_1S_2(a',b')\setminus Z \preceq (x,y)$. In the same way, from
$(a,b)S_1S_2(a',b')\setminus Z \preceq (x,y)$ one gets $(a,b)R_1R_2(a',b')\setminus
Z\preceq (x,y)$. Taken together, this shows $P(R_1R_2,S_1S_2)$.

Suppose $\bigcup \CU \equiv_{n+1} \bigcup \CV$ comes from
$\CU,\CV\in\CR(\CR(K_1\times K_2))$ with $(\CU/_{\equiv})^\da =
(\CV/_{\equiv})^\da$ in stage $n$, i.e.~
\[ \forall R\in \CU\,\exists S\in \CV (R\cup S\equiv_n S)
\wedge \forall S\in\CV\,\exists R\in \CU(S\cup R\equiv_n R),
\]
by the condition to close $\equiv$ under suprema. By induction,
\begin{equation}\label{eqn-sum-equiv}
  \forall R\in \CU\,\exists S\in \CV\,P(R\cup S, S)
\wedge \forall S\in\CV\,\exists R\in \CU\,P(S\cup R,R).
\end{equation}
Let $(a,b),(a',b'),(x,y) \in K_1\times K_2$, and assume $(a,b)(\bigcup \CU)
(a',b')\setminus Z \preceq (x,y)$, i.e.~
\[ \forall R\in \CU((a,b)R(a',b')\setminus Z \preceq (x,y)). \]
To show $(a,b)(\bigcup\CV)(a',b')\setminus Z \preceq (x,y)$, let
$S\in\CV$. By (\ref{eqn-sum-equiv}), there is $R\in\CU$ with $P(S\cup R,R)$, hence
\[ (a,b)(S\cup R)(a',b') \setminus Z \preceq (x,y) \iff (a,b)R(a',b') \setminus Z \preceq (x,y). \]
Since the right-hand side is true, we get $(a,b)S(a',b')\setminus Z \preceq (x,y)$
from the left-hand side. This shows $\forall S\in \CV((a,b)S(a',b')\setminus Z
\preceq (x,y))$, i.e.~$(a,b)(\bigcup\CV)(a',b')\setminus Z \preceq (x,y)$. The
reverse implication
\[ (a,b)(\bigcup\CU)(a',b')\setminus Z \preceq (x,y) \Leftarrow
(a,b)(\bigcup\CV)(a',b')\setminus Z \preceq (x,y)
\]
is shown by a symmetric argument. Therefore, we have $P(\bigcup\CU,\bigcup\CV)$.
\end{proof}

\section*{Acknowledgements}  
We thank the referees for careful reading and many helpful suggestions for improvement.
The first author wishes to acknowledge the support of his family and support and inspiration of Melanie and Lydia.  
The second author thanks the library of the Deutsches Museum in Munich for an
excellent public working environment.

\end{document}